\documentclass[journal]{IEEEtran}
\usepackage{graphics} 
\usepackage{epsfig}
\usepackage{times}
\usepackage{amsmath}
\usepackage{amsthm}
\usepackage{booktabs}
\usepackage{amssymb}
\usepackage{tcolorbox}
\usepackage{multirow}
\usepackage{longtable}
\usepackage {threeparttable}
\usepackage{mathtools}
\usepackage[utf8]{inputenc}
\usepackage{tikz}
\usetikzlibrary{arrows.meta}
\usetikzlibrary{graphs}
\usepackage[colorlinks=false, urlcolor=blue, pdfborder={0 0 0}]{hyperref}
\usepackage{enumerate}
\usepackage{enumitem}
\usepackage{color}
\usepackage[linesnumbered,boxed,ruled,commentsnumbered]{algorithm2e}
\usepackage{subfigure}
\usepackage{comment}
\usepackage{cite}
\usepackage{flushend} 
\usepackage{float} 
\newfloat{tcolorboxenv}{h}{lob}

\definecolor{myred}{RGB}{182,20,50}
\definecolor{myblue}{RGB}{227, 245, 250}
\definecolor{myyellow}{RGB}{255,255,0}
\definecolor{mygreen}{RGB}{250, 250, 235}
\definecolor{myorange}{RGB}{255,128,0}
\definecolor{mygray}{RGB}{192,192,192}
\newtheorem{mydef}{Definition}
\newtheorem{mythm}{Theorem}
\newtheorem{myprob}{Problem}

\newtheorem{mypro}{Proposition}
\newtheorem{mycla}{Claim}
\newtheorem{myexm}{Example}
\newtheorem{remark}{Remark}
\usetikzlibrary{patterns}
\usetikzlibrary{shapes}

\def \M{\mathcal{M}}
\def \S{\mathcal{S}}
\def \A{\mathcal{A}}

\definecolor{myblue}{RGB}{135,206,235}
\definecolor{myyellow}{RGB}{255,255,0}
\definecolor{mygreen}{RGB}{0,255,0}

\title{Optimal Control   of Markov Decision Processes   for Efficiency with Linear Temporal Logic Tasks}

\author{
Yu Chen, Xunyuan Yin, Shaoyuan Li and Xiang Yin%
\thanks{This work was supported by the  National Natural Science Foundation of China (62173226, 62061136004, 61833012).}
	\thanks{Yu Chen, Shaoyuan Li and Xiang Yin are with School of Automation and Intelligent Sensing, Shanghai Jiao Tong University, Shanghai 200240, China.
	{\tt\small \{yuchen26, syli, yinxiang\}@sjtu.edu.cn}.   Xuanyuan Yin is with School of Chemistry, Chemical Engineering and Biotechnology, Nanyang Technological University, Singapore.
    (Corresponding author: Xiang Yin) 
}
}

\IncMargin{0.8em}
\begin{document}

\maketitle

\begin{abstract}
We investigate the problem of optimal control synthesis for Markov Decision Processes (MDPs), addressing both qualitative and quantitative objectives. Specifically, we require the system to satisfy a qualitative task specified by a Linear Temporal Logic (LTL) formula with probability one. Additionally, to quantify the system's performance, we introduce the concept of \emph{efficiency}, defined as the \emph{ratio} between rewards and costs. This measure is more general than the standard long-run average reward metric, as it seeks to maximize the reward obtained \emph{per unit cost}. Our objective is to synthesize a control policy that not only ensures the LTL task is satisfied but also maximizes efficiency. We present an effective approach for synthesizing a stationary control policy that achieves $\epsilon$-optimality by integrating state classifications of MDPs with perturbation analysis in a novel manner. Our results extend existing work on efficiency-optimal control synthesis for MDPs by incorporating qualitative LTL tasks. Case studies in robot task planning are provided to illustrate the proposed algorithm.
\end{abstract}


\IEEEpeerreviewmaketitle

\section{Introduction} 
Decision-making in dynamic environments is a fundamental challenge for autonomous systems, requiring them to react to uncertainties in real-time to achieve desired tasks with performance guarantees. Markov Decision Processes (MDPs) offer a theoretical framework for sequential decision-making by abstracting uncertainties in both environments and system executions as transition probabilities. Leveraging MDPs allows for the analysis of system behavior and the synthesis of optimal control policies through systematic procedures.
In the context of autonomous systems, MDPs have found extensive applications across various domains such as swarm robotics \cite{haksar2023constrained}, autonomous driving \cite{li2019stochastic}, and underwater vehicles \cite{paull2018probabilistic}; 
reader is referred to recent surveys  for additional references and applications \cite{luckcuck2019formal,lauri2022partially,kurniawati2022partially,yin2024formal}.

To assess the performance of infinite horizon behaviors, two widely recognized measures are the \emph{long-run average reward} (or mean payoff) and the \emph{discounted reward} \cite{puterman}. The long-run average reward quantifies the average reward received per state as the system evolves infinitely towards a steady state. However, this measure overlooks the costs incurred for each reward. For instance, a cleaning robot may prioritize collecting more trash while conserving energy. 
Therefore, recently, the notion of \emph{efficiency} has emerged to capture the  \emph{reward-to-cost ratio}  \cite{bloem2014synthesizing,von2016synthesizing}. Specifically, the efficiency of a system trajectory is defined as the ratio between accumulated reward and accumulated cost. The efficient controller synthesis problem thus aims to maximize the expected long-run efficiency \cite{von2016synthesizing,van2021supervisor,van2018performance,lv2024optimal}.

In addition to maximizing quantitative performance measures, many applications require achieving qualitative tasks. Recently, within the context of  MDPs, there has been a growing interest in synthesizing control policies to maximize the probability of satisfying high-level logic tasks expressed, for example, in linear temporal logic (LTL) or omega-regular languages.
For instance, when the MDP model is known precisely, offline algorithms have been proposed to synthesize optimal controllers under LTL specifications; see, e.g., \cite{ding2014optimal,guo2018probabilistic,niu2019optimal,savas2019entropy,xie2021secure,guo2023hierarchical,bals2024multigain}. Recently, reinforcement learning for LTL tasks has also been investigated for MDPs with unknown transition probabilities \cite{hahn2019omega,cai2021modular,wen2021probably,voloshin2022policy,kantaros2024sample}.
As a special instance, the \emph{surveillance task}, which arises in the persistent surveillance of autonomous systems \cite{smith2011optimal,kantaros2020stylus,CHEN20234601}, can also be captured by an LTL task, as it is essentially equivalent to the concept of the B\"{u}chi accepting condition. This condition requires that certain desired target states are visited infinitely often. In general, LTL tasks can capture more complex behaviors and system constraints.

In this work, we investigate the synthesis of control policies for MDPs with both qualitative and quantitative requirements. Specifically, for the qualitative aspect, we require that the LTL task is satisfied with probability one (w.p.1). For the quantitative aspect, we adopt the efficiency measure. Our overarching objective is to maximize the expected long-run efficiency while ensuring the satisfaction of the LTL task w.p.1.
It is worth noting that existing works typically focus on either efficiency optimization (ratio objectives) without qualitative requirements \cite{von2016synthesizing}, or they consider qualitative requirements under the standard long-run average reward (mean payoff) measure \cite{chatterjee2015measuring}. In \cite{ding2014optimal}, the authors consider qualitative requirements expressed by LTL formulas, with a quantitative measure referred to as the \emph{per-cycle} average reward. However, the per-cycle average reward is essentially a special instance of the ratio objective by setting a unit cost for specific states in the denominator.
To the best of our knowledge, the simultaneous maximization of efficiency while achieving the LTL task has not been addressed in the existing literature. This gap motivates our work, where we propose a novel framework for solving MDPs with both qualitative and quantitative objectives, aiming to balance long-run efficiency and the satisfaction of high-level LTL tasks.

To fill this gap in research, we present an effective approach to synthesize stationary policies achieving $\epsilon$-optimality. Our approach integrates state classifications of MDPs \cite{baier2008principles} and perturbation analysis techniques \cite{cao1998relations,xiren2007,cassandras2008introduction} in a novel manner.
Specifically, the key idea of our approach is as follows. Initially, we decompose the MDPs into accepting maximal end components (AMECs) using state classifications, where for each AMEC, we solve the standard efficiency optimization problem without considering the LTL task \cite{von2016synthesizing}. Subsequently, we synthesize a basic policy that achieves optimal efficiency but may fail to fulfill the LTL task. Finally, we \emph{perturb} the basic policy ``slightly" by introducing a target-seeking policy such that the quantitative performance is decreased to 
$\epsilon$-optimal, while still ensuring that the LTL task is fulfilled.
Our approach demonstrates that perturbation analysis is a conceptually simple yet powerful technique for solving MDPs with both qualitative and quantitative tasks, offering new insights into addressing this class of problems. Furthermore, our results also generalize existing results on perturbation analysis from long-run average reward optimizations to the case of long-run efficiency optimizations. This extension opens up new possibilities for applying perturbation analysis to more complex decision-making scenarios involving both qualitative tasks (such as LTL specifications) and quantitative objectives (such as efficiency maximization).

The rest of the paper is organized as follows. 
In Section II, we present some necessary backgrounds and notations. 
Then, we formulate the efficiency optimization problem under LTL tasks in Section III. 
In Section IV, we solve the problem for the special case of communicating MDPs based on a new result from perturbation analysis. 
The general case of non-communicating MDPs is tackled in Section V. 
Case studies of robot task planning are provided in Section VI. 
Finally, we conclude the paper in Section VII.
A preliminary and partial version of this paper was presented in \cite{chen2024optimal}. 
Compared with the conference version, the present journal version has the following main differences. 
First, this paper considers the general LTL task, while \cite{chen2024optimal} only considers the surveillance task, which is a special instance.
Second, we provide  rigorous proofs that cover  the structural properties of this problem and the existence of an optimal solution. Furthermore, we provide extensive case studies and simulations to illustrate the effectiveness of the proposed method.

\section{Preliminary}\label{sec:prelinimary}
\subsection{Markov Decision Processes}
\begin{mydef}[\bf Markov Decision Processes]\upshape
A (finite and labeled) Markov decision process (MDP) is a $6$-tuple
\[
\mathcal{M} = (S,s_0,A,P,\mathcal{AP},\ell),
\]
 where 
 $S$ is a finite states set, $s_0 \in S$ is the initial state,  $A$ is a finite actions set, 
 $P : S \times A \times S \rightarrow [0,1]$ is a transition function such that
  $\forall s \in S, a \in A: \sum_{s' \in S} P(s'\mid s,a)\in \{0,1\}$, 
  $\mathcal{AP}$ is a atomic propositions set, and
 $\ell:S \to 2^{\mathcal{AP}}$ is a labeling function assigning each state a set of atomic propositions.
\end{mydef}

 We also write $ P(s'\mid s,a)$ as $P_{s,a,s'}$.  
 For $s\in S$, the available actions set at $s$ is defined by $A(s)=\{a\in A: \sum_{s' \in S} P_{s,a,s'} = 1\}$. 
 We assume that each state has at least one available action, i.e., $\forall s\in S: A(s)\neq \emptyset$. 
 An MDP induces a directed graph (digraph) such that each vertex is a state and an edge of form $\langle s,s'\rangle$ is defined if $P_{s,a,s'}>0$ for some $a\in A(s)$. Given an MDP $ \mathcal{M}$,    
a \emph{sub-MDP} is a tuple $(\mathcal{S},\mathcal{A})$ such that $\emptyset \neq \mathcal{S} \subseteq S$ is a states subset and 
$\mathcal{A} :\mathcal{S} \rightarrow 2^{A}\setminus \emptyset$ is a function satisfying
 (i) $\forall s\in \mathcal{S}: \mathcal{A}(s) \subseteq A(s)$; and 
(ii) $\forall s\in \mathcal{S},a\in \mathcal{A}(s): \sum_{s' \in \S} P_{s,a,s'}=1$. 
Essentially,  $(\mathcal{S},\mathcal{A})$ induces a new MDP  by restricting the state space to $\mathcal{S}$ and available actions to $\mathcal{A}(s)$ for each state $s\in \mathcal{S}$. 

\begin{mydef}[\bf Maximal End Components]\upshape
Let  $(\mathcal{S},\mathcal{A})$ be a sub-MDP of $\mathcal{M}= (S,s_0,A,P,\mathcal{AP},\ell)$.
$(\mathcal{S},\mathcal{A})$ is said to be an \emph{end component} (EC) if its induced digraph is strongly connected.  
We say $(\mathcal{S},\mathcal{A})$ is a \emph{maximal end component} (MEC) if it is an EC and there is no other end component $(\mathcal{S}',\mathcal{A}')$  such that 
(i) $\mathcal{S}\subseteq \mathcal{S}'$; and (ii) $\forall s \in \S, \A(s) \subseteq \A'(s)$.
We denote by $\texttt{MEC}(\mathcal{M})$ the MECs set of $\mathcal{M}$.
\end{mydef} 
Intuitively, if $(\mathcal{S},\mathcal{A})$ is an MEC, then we can find a policy such that, once $\S$ is reached, we will stay in the MEC forever and all states in $\S$ will be visited infinitely w.p.1 thereafter.

 A Markov chain (MC) $\mathcal{C}$ is an MDP such that $|A(s)|=1$ for any $s \in S$. 
We denote by $\mathbb{P} \in \mathbb{R}^{|S|\times |S|}$ the transition matrix of MC, i.e.,  $\mathbb{P}_{s,s'}=P(s'\mid s,a)$, where $a\in A(s)$ is the unique action at state $s$.  
 Therefore, we can omit actions set of MC and write it as $\mathcal{C}=(S,s_0,\mathbb{P})$.
 The \emph{limit transition matrix} of MC is defined by
 $\mathbb{P}^\star=\lim_{n \rightarrow \infty} \frac{1}{n} \sum_{k=0}^{n-1} \mathbb{P}^{k}$, which always exists for finite MC \cite{puterman}. Let $\pi_{0} \in \mathbb{R}^{|S|}$ be the \emph{initial distribution} where $\pi_{0}(s)=1$ if $s$ is initial state and $\pi_{0}(s)=0$ otherwise.
 A state is said to be  \emph{transient} if its corresponding column in the limit transition matrix is a zero vector; otherwise, the state is \emph{recurrent}.

For $t=0,1,\dots$, we define the history set up to time instant $t$ recursively by $H_0=S$ and when $t\geq 1$, $H_t=H_{t-1}\times A \times S$. A \emph{policy} for an MDP $\mathcal{M}$ is a sequence 
$\mu = (\mu_{0}, \mu_{1},... )$, where  $\mu_{t} : H_t \times A \rightarrow [0,1]$ satisfies $\forall h_t=s_0a_0\dots s_t \in H_t: \sum_{a \in A(s_t)} \mu_{k}(h_t,a)=1$. A policy $\mu = (\mu_{0}, \mu_{1},... )$ is said to be \emph{stationary} if the decision rules are state-based and same at each time instant, i.e., $\forall i,\mu_{i}=\mu'$ such that $\mu': S \times A \to [0,1]$ satisfies $\forall s \in S: \sum_{a \in A(s)}\mu'(s,a)=1$. We write a stationary policy as $\mu=(\mu,\mu,\dots)$ for simplicity. 
Given an MDP $\mathcal{M}$, the sets of all policies and all stationary policies are denoted by $\Pi_\mathcal{M}$ and $\Pi^{S}_\mathcal{M}$, respectively. 
For policy $\mu \in \Pi_{\M}^S$, it induces a transition matrix $\mathbb{P}^{\mu}$, 
where $\mathbb{P}^{\mu}_{i,j}=\sum_{a\in A(i)}\mu(i,a)P_{i,a,j}$.  

Let $\Omega=(S\times A)^\infty$ be the sample space of the MDP and $X_t$, $Y_t$ be the random variables such that $X_t(w)=s_t$ and $Y_t(w)=a_t$ for $w=s_0a_0s_1a_1\dots \in \Omega$. Define the history process $Z_t$ by $Z_t(w)=(s_0,a_0,s_1,a_1,\dots,s_t)$. A policy $\mu=(\mu_0,\mu_1,\dots)\in \Pi_\M$ induces a probability measure $\textsf{Pr}^\mu_\M$ s.t.
\begin{align}
    &\textsf{Pr}^\mu_\M(X_0=s)=\pi_0(s) \nonumber \\
    &\textsf{Pr}^\mu_\M(Y_t=a \mid Z_t=h_t)=\mu_t(h_t,a) \nonumber \\
   & \textsf{Pr}^\mu_\M(X_{t+1}=s'\mid Z_t=(h_{t-1},a,s),Y_t=a_t)=P(s'\mid s,a_t) \nonumber
\end{align}
where $\pi_0$ is initial distribution, $h_t \in H_t$ is a history up to time $t$. Readers can find detailed information about this standard probability measure in \cite{puterman}.

For $\omega=s_0a_0s_1a_1\dots \in \Omega$, the \emph{limit} of $\omega$, denoted by $\texttt{limit}(\omega)$, is the state action pair $(\S_\omega,\A_\omega)$ such that $\S_\omega \subseteq S$ is  the set of states that are visited infinitely often in $\omega$ and $\A_\omega: S_\omega \to 2^A$ is the set of actions chosen infinitely often, i.e.,
\[
  \A_\omega(s)=\{ a \in A(s) \mid \forall m, \exists n>m , \text{ s.t. } s_n=s, a_n=a \}.
\]
 For $\mu \in \Pi_{\M}$ and MEC $(\S,\A) \in \texttt{MEC}(\M)$, 
let
\begin{equation} \label{eq:stayingforeverinMEC}
    \textsf{Pr}^\mu_\texttt{R}(\S,\A)=\textsf{Pr}^{\mu}_{\M}(\{ \omega \in \Omega \mid \texttt{limit}(\omega) = (\Tilde{\S},\Tilde{\A}), \Tilde{\S} \subseteq \S \})
\end{equation}
be probability of staying forever in MEC $(\S,\A)$.
\subsection{Ratio Objectives for Efficiency}
In the context of MDPs, 
quantitative measures such as \emph{average reward}  have been widely used for systems operating in infinite horizon. 
In \cite{bloem2014synthesizing,von2016synthesizing}, a general quantitative measure called \emph{ratio objective} is proposed to characterize the \emph{efficiency} of policies. 
Specifically, two different functions are involved: 
\begin{itemize}
\item 
a \emph{reward function} $\mathtt{R}: S\times A \to \mathbb{R}$ assigning each state-action pair a reward; and 
\item 
a  \emph{cost function} $\mathtt{C}: S \times A \to \mathbb{R}_+$ assigning each state-action pair a positive cost.
\end{itemize}
Then the \emph{efficiency value}  from  initial state $s_0$ under  policy $\mu \in \Pi_{\M}$ w.r.t.\ reward-cost pair $(\mathtt{R},\mathtt{C})$ is defined by
\begin{equation} \label{eq:ratioobjectdef}
    J^{\mu}(s_0, \mathtt{R}, \mathtt{C}) := \liminf_{N\to +\infty} E \left\{  \frac{\sum_{i=0}^{N} \mathtt{R}(s_i,a_i)}{\sum_{i=0}^{N} \mathtt{C}(s_i,a_i)} \right\}, 
\end{equation}
where $E\left\{ \cdot \right\}$ is the expectation of probability measure $\textsf{Pr}_\M^\mu$. 
We omit the reward and cost functions if they are clear by context.
Intuitively, $J^{\mu}(s_0)$ captures the average reward the system received \emph{per cost}, i.e., the efficiency. 
Let $\Pi \subseteq \Pi_\M$  be a set of policies. 
Then optimal efficiency value among policy set $\Pi$  is denoted by $J(s_0, \Pi) = \sup_{\mu \in \Pi} J^\mu(s_0)$. A policy $\mu \in \Pi_\M$ is \emph{optimal} (respectively, \emph{$\epsilon$-optimal}) among policies set $\Pi$ if for all $s \in 
S$, we have $J^\mu(s)=J(s,\Pi)$ (respectively,  $J^\mu(s)\geq J(s,\Pi)-\epsilon$).
Note that the standard long-run average reward is a special case of ratio objective by taking  $\mathtt{C}(s,a)=1, \forall s\in S,a\in A(s)$.
For this case, 
we denote by 
$W^{\mu}(s_0,\mathtt{R}):=J^{\mu}(s_0,\mathtt{R},\mathbf{1})$ the standard long-run average reward
from  initial state $s_0$ under policy $\mu$, 
and denote by $W(s_0, \Pi) = \sup_{\mu \in \Pi} W^\mu(s_0)$ the optimal long-run average reward among policies set $\Pi$.

\subsection{Linear Temporal Logic}
Let $\mathcal{AP}$ be the atomic propositions set. We express formal tasks by Linear Temporal Logic (LTL), which is constructed based on atomic propositions, Boolean operators and temporal operators. Specifically, the syntax of LTL formulae is defined recursively as follows: 
\[
\varphi 
::= 
true  
\mid a 
\mid \varphi_1\wedge \varphi_2
\mid \neg \varphi
\mid \bigcirc \varphi
\mid \varphi_1 U \varphi_2,
\] 
where $a \in \mathcal{AP}$ is an atomic proposition; 
$\neg$ and $\wedge$ are Boolean operators ``negation" and ``conjunction", respectively; 
$\bigcirc$ and $U$ are temporal operators  ``next'' and ``until'', respectively. 
Note that one can further induce  temporal operators 
such as  ``eventually''  $\lozenge \varphi := true U \varphi$ 
and 
``always''  $\square\varphi := \neg \lozenge \neg \varphi$. 

An LTL formula $\varphi$ is interpreted over infinite words on $2^{\mathcal{AP}}$. 
Readers are referred to \cite{baier2008principles} for details on semantics of LTL formulae.
For infinite word $\sigma \in (2^{\mathcal{AP}})^{\infty}$, 
we denote by $\sigma \models \varphi$ if it satisfies LTL formula $\varphi$. 
The set of all infinite words satisfying  $\varphi$ is denoted by $\mathcal{L}_{\varphi}=\{ \sigma \in (2^{\mathcal{AP}})^{\infty}\mid 
\sigma \models \varphi \}$. 

\begin{mydef}[\bf Deterministic Rabin Automata]\upshape
A \emph{deterministic Rabin automata} (DRA) is a tuple $R=(Q,q_0,\Sigma,\delta,Acc)$, where $Q$ is a finite states set, $q_0 \in Q$ is the initial state, $\Sigma$ is a finite alphabet set, $\delta: Q \times \Sigma \to Q$ is the transition function, and $Acc=\{(B_1,G_1),\dots,(B_n,G_n) \}$ is a finite set of Rabin pairs such that $B_i,G_i \subseteq Q$ for all $i=1,2,\dots,n$.
\end{mydef}

For an infinite word $\sigma=\sigma_1\sigma_2\cdots \in \Sigma^{\infty}$,  
its induced infinite \emph{run} in DRA $R$ is the sequence of states $\rho=q_0 q_1\cdots \in Q^{\infty}$ such that $q_i=\delta(q_{i-1},\sigma_i)$ for all $i \geq 1$. An infinite run $\rho$ is said to be accepted if  there exists a Rabin pair $(B_i,G_i) \in Acc$ such that $\textsf{inf}(\rho)\cap G_i\neq \emptyset$ and $\textsf{inf}(\rho) \cap B_i = \emptyset$, where $\textsf{inf}(\rho)$ is the set of states that occur infinitely many times in $\rho$. 
An infinite word $\sigma$ is said to be \emph{accepted} if its induced  infinite run is accepted. 
We denote by $\mathcal{L}(R) \subseteq \Sigma^{\infty}$ the set of all accepted words of DRA $R$. 
For an arbitrary LTL formula $\varphi$ over $\mathcal{AP}$, 
it is well-known that \cite{baier2008principles}, there exists a DRA with $\Sigma=2^{\mathcal{AP}}$ that accepts all infinite words satisfying $\varphi$, i.e., $\mathcal{L}_{\varphi}=\mathcal{L}(R)$.

For an MDP $\M$, a sample path $\omega=s_0a_0s_1a_1\cdots \in \Omega$   generates a word $\ell(\omega)=\ell(s_0)\ell(s_1)\cdots  \in (2^\mathcal{AP})^\infty$. 
Given an LTL formula $\varphi$ and a policy $\mu\in \Pi_\M$, we define
\[
\textsf{Pr}^{\mu}_{\M}(s_0 \models \varphi) :=\textsf{Pr}^{\mu}_{\M}(\{\omega \in \Omega \mid \ell(\omega) \models \varphi \})
\]
as the probability of satisfying LTL formula $\varphi$ for MDP $\M$ under policy $\mu \in \Pi_{\M}$ initial from $s_0$. 
We denote by $  \Pi^{\varphi}_{\M}$ the set of policies under which the LTL task can be satisfied with probability one, i.e., 
  \[
  \Pi^{\varphi}_{\M}=\{ \mu \in \Pi_{\M} \mid   \textsf{Pr}^{\mu}_{\M}(s_0 \models \varphi)=1\}.
  \]
\subsection{Product MDPs}
 We construct the product system between the original MDP and the DRA representing the LTL task to integrate the task information into the MDP model.
 
\begin{mydef}[\bf Product MDPs]\upshape
Let $\M=(S,s_0,A,P,\mathcal{AP},$ $\ell)$ be an MDP and $R=(Q,q_0,2^{\mathcal{AP}},\delta,Acc)$ be the DRA such that $\mathcal{L}_\varphi=\mathcal{L}(R)$. 
The \emph{product MDP} is a 7-tuples
\[
\M_{\otimes}=(S_{\otimes},s_{0,{\otimes}},A,P_{\otimes},\mathcal{AP},\ell_{\otimes},Acc_{\otimes}),
\]
where $S_{\otimes} = S \times Q$ is the product state space, 
$s_{0,{\otimes}}=(s_0,q)$ is the initial state such that $q=\delta(q_0,\ell(s_0))$,
 $P_{\otimes}:S_{\otimes} \times A \times S_{\otimes} \to [0,1]$ is the transition function defined by
  \begin{align}\label{eq:producttran}
	P_{\otimes}((s,q),a,(s',q')) \!= \!
		\left\{\!\!
		\begin{array}{cl}
			P_{s,a,s'} & \text{if }    q'=\delta(q,\ell(s'))  \\
			0                & \text{otherwise}  
		\end{array}
		\right.\!\!\!,   
\end{align} 
$\ell_{\otimes}$ is the labeling function such that
$\ell_{\otimes}((s,q))=\ell(s)$ and
$Acc_{\otimes}=\{ (B_1^{\otimes},G_1^{\otimes}),\dots,(B_n^{\otimes},G_n^{\otimes}) \}$ such that $B_i^{\otimes}=S \times B_i$ and $G_i^{\otimes}=S \times G_i$ for all $i=1,\dots n$.
\end{mydef}
Note that, since $R$ is deterministic and the action spaces of $\M$ and $\M_{\otimes}$ are same, there exists a one-to-one correspondence between policies in $\M$ and $\M_{\otimes}$ 
\cite{baier2008principles, guo2018probabilistic}. Hereafter in this paper, we will omit the subscript and directly denote by $\M=(S,s_0,A,P,\mathcal{AP},\ell,Acc)$ the product MDP for the sake of simplicity. 
The control synthesis problem is solved based on the product MDP. 
Specifically, the reward and cost functions can be directly defined by first component of product state. Furthermore, for any state sequence $\rho\in S^{\infty}$ in (product) MDP, it satisfies the LTL formula if and only if there exists an accepting pair $(B_{k},G_{k})\in Acc$ such that $\textsf{inf}(\rho)\cap G_{k} \neq \emptyset$ and $\textsf{inf}(\rho)\cap B_{k} = \emptyset$.  This accepting condition can be captured by the notion of maximal accepting end component. 

\begin{mydef}[\bf Maximal Accepting End Components]\upshape
Given a product MDP $\M=(S,s_0,A,P,\mathcal{AP},\ell,Acc)$, an \emph{accepting end component} (AEC) of $\M$ is an EC $(\S,\A)$ such that for some accepting pair $(B_{k},G_{k})\in Acc$, 
we have $\S \cap B_{k} = \emptyset$ and $\S \cap G_{k} \neq \emptyset$. Moreover $(\S,\A)$ is said to be an \emph{maximal accepting end component} (MAEC) 
if there exists no other AEC $(\S',\A')$ such that 
(i) $\mathcal{S}\subseteq \mathcal{S}'$; and (ii) $\forall s\in \S, \A(s) \subseteq \A'(s)$. We denote by $\texttt{AEC}(\M)$ and $\texttt{MAEC}(\M)$ the set of AECs and MAECs of product MDP $\M$, respectively. 
\end{mydef}
Intuitively, for policy $\mu \in \Pi_{\M}$, the probability of satisfying a given LTL formula is equal to the probability of reaching MAEC and staying in there forever. 
Note that both the MECs set and the MAECs set can be computed effectively via graph search over the product state space; see, e.g., \cite{baier2008principles,guo2018probabilistic}. 
For any MAEC $(\S',\A') \in \texttt{MAEC}(\M)$, it is contained in some MEC $(\S,\A) \in \texttt{MEC}(\M)$ such that $\S' \subseteq \S$ and $\forall s \in \S'$, $\A'(s) \subseteq \A(s)$. MEC $(\S,\A) \in \texttt{MEC}(\M)$ is an \emph{accepting maximal end component} (AMEC) if it contains at least one MAEC. We denote by $\texttt{MEC}_{\varphi}(\M)$ the set of AMECs. We use the following example to illustrate notions of different end components.
\begin{myexm}
    Let us consider a product MDP $\M$ shown in Figure~~\ref{fig:exampleend}. For each action, the transition probability is one and the value is omitted in the figure. This MDP has two MECs, i.e., $\texttt{MEC}(\M)=\{ (\S_1,\A_1), (\S_2,\A_2) \}$ such that $\S_1=\{ 2 \}$, $\A_1(2)=\{  a_1\}$ and $\S_2=\{ 3,4\}$, $\A_2(3)=\{a_1\}$, $\A_2(4)=\{ a_1,a_2 \}$. The only accepting pair of $\M$ is $(\{ 3\},\{4 \}) $. Then MDP has one MAEC, i.e., $\texttt{MAEC}(\M)=\{ (\S_3,\A_3) \}$ with $\S_3=\{4\}$ and $\A_3=\{ a_2\}$. Since $(\S_3,\A_3)$ is contained in $(\S_2,\A_2)$, the only AMEC is $(\S_2,\A_2)$, i.e., $\texttt{MEC}_\varphi(\M)=\{ (\S_2,\A_2)\}$.
\end{myexm}

\begin{figure}
    \centering
    \begin{tikzpicture}
[
square/.style={circle, draw=black!255, fill=white!255, very thick, minimum height=5mm,minimum width=5mm},
]
	\node [square](q1)at(0,2){$2$};
	\node [square](q2)at(2,2){$1$};
	\node [square](q3)at(4,2){${3}$};
	\node [square](q4)at(6,2){${4}$};
        \node (a11)at(-0.65,1.6){$a_{1}$};
	\node (a21)at(1,2.2){$a_{1}$};
        \node (a23)at(3,2.2){$a_{2}$};
        \node (a34)at(5,1.8){$a_{1}$};
        \node (a43)at(5,2.2){$a_{1}$};
        \node (a44)at(6.7,1.6){$a_{2}$};
	\draw[->] (q2) -- (q1);
	\draw[->] (q2) -- (q3);

\draw[->] (4.35,1.8)..controls (4.78,1.5) and (5.21,1.5) .. (5.65,1.8);

\draw[->] (5.65,2.2)..controls (5.21,2.5) and (4.78,2.5) .. (4.35,2.2);

	\draw[->] (q4) .. controls +(right:10mm) and +(down:10mm) .. (q4);

 	\draw[->] (q1) .. controls +(left:10mm) and +(down:10mm) .. (q1);
    \end{tikzpicture} 
    \caption{Example to illustrate different end components.}
    \label{fig:exampleend}
\end{figure}

\section{Problem Formulation} \label{sec:problemstatement}
In general, quantitative efficiency cannot precisely capture complex qualitative requirements. As a result, a system optimized purely for efficiency may engage in undesirable or even forbidden behaviors.
In this work, we aim to synthesize a control policy under both performance and correctness considerations such that 
\begin{itemize}
    \item The given LTL task is satisfied with probability 1; and
    \item the efficiency is maximized under LTL task constraint.
\end{itemize}

Now we formulate the problem solved in this paper. 
\begin{myprob}[\bf Efficiency Maximization for Linear Temporal Logic Tasks]\label{problem:ratioltl}
 Given MDP $\mathcal{M}$ and LTL formula $\varphi$, 
 which is equivalent to given the product MDP, reward function $\mathtt{R}$, cost function $\mathtt{C}$ and a threshold value $\epsilon >0$, assme that $\Pi^\varphi_\M \neq \emptyset$.
  Find a stationary policy $\mu^\star \in \Pi^{\varphi}_\mathcal{M} \cap \Pi_{\M}^S$ such that  
  \begin{equation}
  J^{\mu^\star}(s_0) \geq J(s_0, \Pi^{\varphi}_{\M}) -\epsilon.
  \end{equation}
  \end{myprob}
Without loss of generality, we assume that, initial from each state in the product MDP, 
there exists a policy under which the LTL task can be finished w.p.1.  
Otherwise, undesired states can be eliminated by  Algorithm 45 in~\cite{baier2008principles} in polynomial time. 
  
\begin{remark}
Before  proceeding further, we make several comments on the above problem formulation. 
\begin{itemize}[leftmargin=1.5em]
  \item
   First, we seek to find an $\epsilon$-optimal policy $\mu^\star$ among all policies satisfying LTL tasks w.p.1. 
   The motivation for this setting is that in general, to achieve the value $J(s_0, \Pi^{\varphi}_{\M})$, we need to apply an infinite memory policy, which is too expensive to realize in practice.
   One is referred to \cite{chatterjee2015measuring} for this issue when quantitative measure is the long-run average reward, which is a special case of our ratio objective.
  \item   
  Second, we further restrict our attention to stationary policies in $\Pi_{\M}^S$ a priori. 
  We will show in the following result that such a restriction is without loss of generality in the sense that a stationary solution always exists.\medskip
\end{itemize} 
\end{remark}

\begin{mypro} \label{prop:stationaryisenough} 
Let $\M=(S,s_0,A,P,\mathcal{AP},\ell,Acc)$ be the product MDP. It holds that $J(s_0,\Pi^\varphi_\M) = J(s_0, \Pi^\varphi_\M \cap \Pi^S_\M)$.
\end{mypro}
\begin{proof}
Consider the optimal deterministic stationary policy  $\mu^\star \in \Pi^{SD}_{\M}$, which is formally defined in Eq.~\eqref{eq:optimaldeterdef}. 
According to Claim~\ref{cla:regularoptimalpolicy} in  Appendix~\ref{appendix:proof}, we know that $\mu^\star$ is regular.
Let $\texttt{R}(\M) \subseteq \texttt{MEC}(\M)$ be the set of MECs that contain recurrent state in MC $\M^{\mu^\star}$.
By 1) of Claim~\ref{claim:deterpolicyproperty} in Appendix~\ref{appendix:proof}, in MC $\M^{\mu^\star}$, for each $(\S,\A) \in \texttt{R}(\M)$, the only recurrent class of $(\S,\A)$ is in some MAEC, denoted by $\texttt{A}(\S,\A) \in \texttt{MAEC}(\M)$.
We define $\mu_{(\S,\A)}$ the policy over sub-MDP $(\S,\A)$ under which it has only one recurrent class consisting of all states in $\texttt{A}(\S,\A)$.
Policy $\mu_{(\S,\A)}$ exists since $(\S,\A)$ is communicating.
We construct a policy $\mu'$ such that 
\begin{align}
	\mu'(s,a) \!= \!
		\left\{\!\!
		\begin{array}{cl}
			\mu_{(\S,\A)}(s,a)  & \text{if }    s\in \S, a \in \A(s), (\S,\A) \in \texttt{R}(\M)  \\
			\mu^{\star}(s,a)               & \text{otherwise.}  
		\end{array}
		\right.\!\!\! \nonumber
  \end{align}
Now consider policy $\mu^\delta=(1-\delta) \mu^\star+ \delta \mu'$ s.t. $\mu^\delta(s,a)=(1-\delta) \mu^\star(s,a)+\delta \mu'(s,a), \forall s \in S, a \in A(s)$.
It is easy to know that $\mu^\delta \in \Pi^\varphi_\M$ for any $0<\delta \leq 1$. 
By \eqref{eq:MECreachingratio}, for $0 \leq \delta \leq 1$,
\begin{equation} \label{eq:ratiomiddle}
    J^{\mu^\delta}(s_0,\mathtt{R},\mathtt{C})=\sum_{(\S,\A) \in \texttt{R}(\M)}\textsf{Pr}^{\mu^\delta}_\texttt{R}(\S,\A)\frac{W^{\mu^\delta}(s_{(\S,\A)},\mathtt{R})}{W^{\mu^\delta}(s_{(\S,\A)},\mathtt{C})}
\end{equation}
where $\textsf{Pr}^{\mu^\delta}_\texttt{R}(\S,\A)$ is probability of staying forever in MEC $(\S,\A)$ defined in \eqref{eq:stayingforeverinMEC} and $s_{(\S,\A)} \in \S$ can be any state in $\S$. By \eqref{eq:samestayingprobability}, $\textsf{Pr}^{\mu^\delta}_\texttt{R}(\S,\A)$ is constant for $\delta \in [0,1]$. From~\cite{cao1998relations} we know that 
$W^{\mu^\delta}(s_{(\S,\A)},(\cdot))$ is continuous w.r.t. $\delta \in [0,1]$ for $(\S,\A) \in \texttt{R}(\M)$ and $(\cdot) \in \{ \mathtt{R}, \mathtt{C} \}$. Thus $J^{\mu^\delta}(s_0,\mathtt{R},\mathtt{C})$ is continuous w.r.t. $\delta \in [0,1]$. 
Then for any $\epsilon >0$, we can find some $\delta > 0$ such that 
\begin{equation}\label{eq:arbitraryclose}
    |J^{\mu^{\delta}}(s_0,\mathtt{R},\mathtt{C})-J^{\mu^\star}(s_0,\mathtt{R},\mathtt{C})| \leq \epsilon.
\end{equation}
Since $\mu^{\delta} \in \Pi^\varphi_\M$ and $\epsilon>0$ is arbitrary, we know that
\begin{align}
    &J(s_0,\mathtt{R},\mathtt{C},\Pi_\M^\varphi) 
        \nonumber\\
       \geq  &J^{\mu^\star}(s_0,\mathtt{R},\mathtt{C}) 
       =  J^{\mu^\star}(s_0,\hat{\mathtt{R}},\mathtt{C})
       = J(s_0,\hat{\mathtt{R}},\mathtt{C},\Pi_\M).
        \nonumber
\end{align}
The first equality comes from 2) of Claim~\ref{claim:deterpolicyproperty} and second equality holds from \eqref{eq:optimaldeterdef}. With \eqref{eq:stationaryhalf}, we have proven that 
\begin{equation}
J(s_0,\hat{\mathtt{R}},\mathtt{C},\Pi_\M)=J(s_0,\mathtt{R},\mathtt{C},\Pi_\M^\varphi)=J^{\mu^\star}(s_0,\mathtt{R},\mathtt{C}).
\end{equation}
From \eqref{eq:arbitraryclose}, policy $\mu^\delta \in \Pi^\varphi_\M \cap \Pi^S_\M$ can achieve $\epsilon$-optimality by picking proper $\delta$ for any $\epsilon >0$. This completes the proof. \nonumber
\end{proof}

\section{Case of Communicating MDPs} \label{sec:solutionforcommunicating}
Before handling the general case, 
in this section, we consider a special case, where the MDP is communicating. 
Formally, an MDP $\mathcal{M}$ is said to be \emph{communicating} if
  \begin{equation} \label{eq:communicating definition}
      \forall s,s'\in S, \exists  \mu \in \Pi_{\mathcal{M}}^S,\exists n\geq 0:  (\mathbb{P}^{\mu})^{n}_{s,s'}>0.
  \end{equation}
In other words, for a communicating MDP, one state can reach another state under some policy.

\textbf{General Idea: } 
When the MDP is communicating, we solve problem~\ref{problem:ratioltl} by the following steps: 
\begin{itemize}[leftmargin=1.5em]
\item First, we compute set $\texttt{MAEC}(\M)$ of all MAECs and handle each sub-MDP $(\S,\A) \in \texttt{MAEC}(\M)$ individually.
\item Then, for each  $(\S,\A)$, we  solve Problem~\ref{problem:ratioltl} by the following two step:
\begin{itemize}
    \item 
    We first find two policies, denoted by $\mu_{opt}$ and $\mu_{irr}$, which maximizes efficiency without considering the LTL task and ensures all states in $\S$ can be visited infinitely often w.p.1, respectively. The discussion on constructing these policies is presented in Section~\ref{sec:review}.
  \item 
    We then \emph{perturb} policy $\mu_{opt}$ ``slightly" by $\mu_{irr}$ 
  such that the efficiency value of the resulting policy is $\epsilon$-close to that of $\mu_{opt}$, and
  the LTL task can still be achieved due to the presence of perturbation $\mu_{irr}$. The $\epsilon$-optimality of perturbed policy is guaranteed by analysis in Section~\ref{sec:pert}.
\end{itemize}
  \item Finally, for entire communicating MDP $\M$, we synthesize a policy under which it will stay in MAEC achieving highest efficiency value among sub-MDPs in $\texttt{MAEC}(\M)$ forever w.p.1. The Algorithm~\ref{alg:cmdpsolution} in Section~\ref{sec:coverall} formally states overall idea.
\end{itemize}
Now, we proceed the above idea in more detail. 

\subsection{Maximum Efficiency Policy and Irreducible Policy}\label{sec:review}
In this subsection, we discuss how to construct the maximum efficiency policy and irreducible policy, which are key components for solving Problem~\ref{problem:ratioltl}. We first review the existing solution for efficiency optimization. 
It has been shown in  \cite{von2016synthesizing} that, for communicating MDP $\mathcal{M}$, 
there exists a stationary policy $\mu \in \Pi^S_{\M}$ such that $J^\mu(s_0)=J(s_0,\Pi_\M)$ and the induced MC $\M^\mu$ is an \emph{unichain} (MC with a single recurrent class and some transient states).  
Furthermore, we have 
  \begin{equation} \label{eq:ratioobjective}
        J^{\mu}(s_0)= \frac{\sum_{s \in S} \sum_{a \in A(s)} \pi(s) \mu(s,a) \mathtt{R}(s,a)}{\sum_{s \in S} \sum_{a \in A(s)} \pi(s) \mu(s,a) \mathtt{C}(s,a)}, 
  \end{equation}
such that $\pi \in \mathbb{R}^{|S|}$ is the unique stationary distribution with $\pi \mathbb{P}^{\mu} = \pi$. 
With this structural property for communicating MDP,  \cite{von2016synthesizing} transforms the policy synthesis problem for efficiency optimization to a parameter synthesis problem described by the nonlinear program~\eqref{opt1:obj}-\eqref{opt1:con5} as follows:
\begin{align}
 &\max_{\gamma(s,a)} \quad \frac{\sum_{s \in S} \sum_{a \in A(s)} \gamma(s,a) \mathtt{R}(s,a)}{\sum_{s \in S} \sum_{a \in A(s)}  \gamma(s,a) \mathtt{C}(s,a)}  \label{opt1:obj} \\
\!\!\!\!\!\!\!\!\text{s.t. } \ \
&q(s,t) = \sum_{a \in A(s)} \gamma(s,a)P(t \mid s,a),\forall s, t \in S   \label{opt1:con1}\\
&\lambda(s) = \sum_{a \in A(s)} \gamma(s,a), \forall s \in S  \label{opt1:con2}\\    
&\lambda(t) = \sum_{s \in S} q(s,t),\forall t \in S  \label{opt1:con3} \\
&\sum_{s \in S} \lambda(s)=1     \label{opt1:con4}\\
&\gamma(s,a) \geq 0, \forall s \in S ,\forall a \in A(s)   \label{opt1:con5}
\end{align} 
Since we will only leverage this existing result, the reader is referred to \cite{von2016synthesizing} for more details on the intuition of the above nonlinear program. The only point we would like to emphasize is that this nonlinear program is a linear-fractional programming, which  can be solved efficiently by converting to a linear program by Charnes-Cooper transformation~\cite{zionts1968programming}.  
Now, let   $\gamma^{\star}(s,a)$ be the solution to  Equations~\eqref{opt1:obj}-\eqref{opt1:con5}. The \emph{optimal policy}, denoted by $\mu_{opt}$,  can be decoded as follows. 
Let $Q=\{ s \in S \mid \sum_{a \in A(s)} \gamma^{\star}(s,a)>0 \}$ and we define
\begin{equation}\label{eq:optcmdppolicy}
    \mu_{opt}(s,a)=\frac{\gamma^{\star}(s,a)}{\sum_{a \in A(s)}\gamma^{\star}(s,a)}, \quad s \in Q.
\end{equation}
For the remaining part, policy $\mu_{opt}$ only needs to ensure that states in $S\setminus Q$ will reach $Q$ eventually w.p.1 in MC $\M^{\mu_{opt}}$; see, e.g., procedure in \cite[Page 480]{puterman}. 
Then such a policy $\mu_{opt}$ achieves $ J^{\mu_{opt}}(s_0)=J(s_0,\Pi_\M)$.
Since the definition of efficiency in \eqref{eq:ratioobjectdef} is slightly different from that in \cite{von2016synthesizing}, we also prove the existence of stationary optimal efficiency policy in Claim~\ref{cla:ratiooptpolicyexist} of Appendix~\ref{appendix:proof} for completeness.

Note that, under policy $\mu_{opt}$, only states in $Q$ will be visited infinitely often, which has no guarantee on satisfaction of LTL task. To this end, we consider an arbitrary stationary policy  $\mu_{irr}\in \Pi^S_\M$, 
which is referred to as the \emph{irreducible policy}, 
such that  $\M^{\mu_{irr}}$ is irreducible. For policy $\mu_{irr}$, we have 
\begin{itemize}[leftmargin=1.5em]
  \item
  It is well-defined since we already assume that the MDP $\M$ is communicating.
  For example, one can simply use the uniform policy as $\mu_{irr}$, i.e., 
   each available action is enabled with the same probability at each state;
  \item 
 When applying irreducible policy over MAEC, all states in MAEC can be visited infinitely often w.p.1. Then from definition of MAEC, it can finish LTL task w.p.1.
\end{itemize}
\subsection{Perturbation Analysis for Efficiency}\label{sec:pert}
Here, we analyse the ratio objective efficiency performance under perturbation, which is used to ensure $\epsilon$-optimality under LTL task constraint for communicating MDP. To this end, we adopt the idea of perturbation analysis of MDP, 
which is originally developed to quantify the difference of long-run average rewards  between two policies \cite{cao1998relations}. 
First, we introduce some related definitions.  

\begin{mydef}[\bf Utility Vectors \& Potential Vectors]\upshape
Let $\mu\in \Pi_\M^S$ be a stationary policy and  $\mathtt{V}: S \times A \to \mathbb{R}$ be a generic utility function,  which can be either the reward function $\mathtt{R}$ or the 
cost function $\mathtt{C}$. 
Then 
\begin{itemize}[leftmargin=1.5em]
\item 
the \emph{utility vector} of  policy $\mu$ (w.r.t.\ utility function $\mathtt{V}$), 
denoted by $v^{\mu}_\mathtt{V} \in \mathbb{R}^{|S|}$, is defined by
\begin{equation}
v^{\mu }_\mathtt{V}(s)=\sum_{a \in A(s)} \mu(s,a) \mathtt{V}(s,a). 
\end{equation}
\item 
the \emph{potential vector} of  policy $\mu$ (w.r.t.\ utility function $\mathtt{V}$), 
denoted by $g^{\mu}_\mathtt{V} \in \mathbb{R}^{|S|}$, is defined by
\begin{equation}\label{eq:potentialvector}
g_\mathtt{V}^{\mu} =(I-\mathbb{P}^{\mu}+(\mathbb{P}^{\mu})^\star)^{-1} v^{\mu}_\mathtt{V}. \medskip \medskip
\end{equation}
\end{itemize} 
\end{mydef}

In the above definition, the potential vector is well-defined as  matrix $I-\mathbb{P}^{\mu}+(\mathbb{P}^{\mu})^\star$ is always invertible \cite{puterman}, where $(\mathbb{P}^{\mu})^\star$ is the limit transition matrix of $\mathbb{P}^{\mu}$. 
Intuitively, 
the potential vector $g^{\mu}_\mathtt{V}$ contains the information regarding the long run average utility in MC $\M^\mu$. Specifically, let $\pi_0$ be the initial distribution and $\pi_\mu$ be the limit distribution such that $\pi_\mu = \pi_0 (\mathbb{P}^{\mu})^\star$.
Then we have   
\[
\pi_\mu^{\top} g^{\mu}_\mathtt{V}=\pi_\mu^{\top}  v^{\mu}_\mathtt{V}=W^\mu(s_0,\mathtt{V}),
\]
which computes the long run average utility under $\mu$. 
Next, we define notion of deviation vectors of two different policies. 

\begin{mydef}[\bf Deviation Vectors]\upshape
Let $\mu, \mu' \in \Pi^S_\M$ be two stationary policies and $\mathtt{V}: S \times A \to \mathbb{R}$ be a   utility function. 
Then the \emph{deviation vector} from $\mu$ to $\mu'$ (w.r.t.\ utility function $\mathtt{V}$) is defined by 
\begin{equation} \label{eq:deviationvector}
\mathbf{D}_\mathtt{V}(\mu,\mu') = (v^{\mu' }_\mathtt{V}-v^{\mu }_\mathtt{V})+(\mathbb{P}^{\mu'}-\mathbb{P}^{\mu})g_\mathtt{V}^{\mu}.\medskip
\end{equation}
\end{mydef}

The deviation vector can be used to compute the difference between the long-run average utility of the original policy and the perturbed policy. 
Formally, let $\mu, \mu' \in \Pi^S_\M$ be two stationary policies, $\mathtt{V}: S \times A \to \mathbb{R}$ be a   utility function and $\delta\in (0,1)$ be the perturbation degree.
We define 
\[
\mu_{\delta}=(1-\delta)\mu+\delta \mu'
\]
s.t. $\mu_\delta(s,a)=(1-\delta) \mu(s,a)+\delta \mu'(s,a), \forall s \in S, a \in A(s)$. 
It was shown in \cite{cao1998relations} that, when $\M^{\mu}$ is a unichain, the differences between 
the long run average utilities of the perturbed policy and the original policy can be calculated as follow:  
\begin{equation} \label{eq:rewarddeviation}
    W^{\mu_\delta}(s_0,\mathtt{V})-W^\mu(s_0,\mathtt{V})
=   \pi_{\mu_{\delta}}^\top v_\mathtt{V}^{\mu_{\delta}} - \pi_{\mu}^\top v_\mathtt{V}^{\mu}
= \delta \pi_{\mu_{\delta}}^\top \textbf{D}_{\mathtt{V}}(\mu,\mu').
\end{equation}

However, the above classical result can only be applied to the case of long-run average reward.  
The following proposition provides the key result of this subsection, 
which shows how to generalize Equation~\eqref{eq:rewarddeviation} 
from long-run average reward to the case of long-run efficiency under the ratio objective. 

\begin{mypro} \label{prop:ratiodeviationequation}
Let $\mu, \mu' \in \Pi^S_\M$ be two stationary policies, 
 $\mathtt{R}: S\times A \to \mathbb{R}$ be the reward function, 
 $\mathtt{C}: S \times A \to \mathbb{R}_+$ be the cost function, 
 and $\delta\in (0,1)$ be the perturbation degree.
Let   $\mu_{\delta}=(1-\delta)\mu+\delta \mu'$ be the perturbed policy. 
If $\M^\mu$ is unichain, then we have
\begin{align}\label{eq:new}
    &J^{\mu_{\delta}}(s_0, \mathtt{R}, \mathtt{C})-J^{\mu}(s_0, \mathtt{R}, \mathtt{C}) \\
= &\frac{\delta}{  \pi_{\mu_{\delta}}^\top v_\mathtt{C}^{\mu_{\delta}}     }    \pi_{\mu_{\delta}}^{\top}
\left( \textbf{D}_{\mathtt{R}}(\mu,\mu') -J^{\mu}(s_0, \mathtt{R}, \mathtt{C})\textbf{D}_{\mathtt{C}}(\mu,\mu')\right).\nonumber \medskip \medskip
 \end{align}
 \end{mypro}
\begin{proof}
First, we prove that the perturbed policy $\mu_\delta$ induces unichain MC. Note that if a state $s$ can reach $s'$ in either $\M^{\mu}$ or $\M^{\mu'}$, then $s$ can reach $s'$ in MC $\M^{\mu_\delta}$. Let $R_{\mu}\subseteq S$ be the unique recurrent class in unichain MC $\M^{\mu}$. Then in MC $\M^{\mu}$, all states in $S$ can reach states in $R_\mu$, which also holds for MC $\M^{\mu_\delta}$. Since for any recurrent class of an MC, all states in the recurrent class can reach each other and can not reach states not in this recurrent class, we can prove by contradiction that for any recurrent class $R_{\mu_\delta} \subseteq S$ in MC $\M^{\mu_\delta}$, $R_{\mu} \subseteq R_{\mu_\delta}$. Then it is easy to know that MC $\M^{\mu_\delta}$ only contains one recurrent class, i.e. $\mu_\delta$ induces unichain MC.

Then we have the following equalities 
\begin{equation}
    \begin{aligned}
        & J^{\mu_{\delta}}(s_0)-J^{\mu}(s_0) \\
        =& \frac{\pi_{\mu_{\delta}}^{\top} v_\mathtt{R}^{\mu_{\delta}}}{\pi_{\mu_{\delta}}^{\top} v_\mathtt{C}^{\mu_{\delta}}} - \frac{J^{\mu}(s_0)\pi_{\mu_{\delta}}^{\top} v_\mathtt{C}^{\mu_{\delta}}}{\pi_{\mu_{\delta}}^{\top} v_\mathtt{C}^{\mu_{\delta}}} \\
        = &   \frac{\pi_{\mu_{\delta}}^{\top} v_\mathtt{R}^{\mu_{\delta}}-\pi_{\mu}^{\top} v_\mathtt{R}^{\mu} -J^{\mu}(s_0)(\pi_{\mu_{\delta}}^{\top} v_\mathtt{C}^{\mu_{\delta}}- \pi_{\mu}^{\top} v_\mathtt{C}^{\mu} )}{\pi_{\mu_{\delta}}^{\top} v_\mathtt{C}^{\mu_{\delta}}} \\
        = & \frac{\delta}{\pi_{\mu_{\delta}}^{\top} v_\mathtt{C}^{\mu_{\delta}}}\pi_{\mu_{\delta}}^{\top}(\textbf{D}_{\mathtt{R}}(\mu,\mu') -J^{\mu}(s_0)\textbf{D}_{\mathtt{C}}(\mu,\mu')). 
        \nonumber
    \end{aligned}
\end{equation}
Specifically,  the first and the second equalities hold because $\mu_{\zeta}$ and $\mu$ induce unichain MCs and the efficiency values can be computed by Equation~\eqref{eq:ratioobjective}. 
The last equality comes from Equation~\eqref{eq:rewarddeviation}. This completes the proof.
 \end{proof}
\begin{remark}
Our new result in Equation~\eqref{eq:new} for ratio objective subsumes 
the classical result in Equation~\eqref{eq:rewarddeviation} for the case of long-run average reward. 
Specifically, 
when $\mathtt{C}(s,a)=1, \forall s\in S,a\in A(s)$, 
$J^{\mu}(s_0, \mathtt{R}, \mathtt{C})$ reduces to $W^{\mu}(s_0, \mathtt{R})$. 
For this case, we know that 
$ \pi_{\mu_{\delta}}^\top v_\mathtt{C}^{\mu_{\delta}}=1$ as $v^{\mu }_\mathtt{C}(s)=1,\forall s\in S$. 
Furthermore, we have $\textbf{D}_{\mathtt{C}}(\mu,\mu')=0$ as both policies achieve the same cost. 
Therefore, Equation~\eqref{eq:new} becomes to Equation~\eqref{eq:rewarddeviation} 
and our result provides a more general form of perturbation analysis in terms of deviation vectors.  
\end{remark}

 \subsection{Efficiency Optimization with LTL Tasks over MAEC} \label{sec:ltlMAEC}
 In this subsection, we assume that the communicating MDP $\M$ is an MAEC, i.e., there exists an accepting pair $(B,G) \in Acc$ such that $S \cap B = \emptyset$ and $S \cap G \neq \emptyset$.
Let $\mu_{opt}$ and $\mu_{irr}$ be the maximum efficiency policy and irreducible policy of $\M$ in Section~\ref{sec:review}, respectively.
We perturb the policy $\mu_{opt}$ by the policy $\mu_{irr}$ to obtain a new policy
\begin{equation} \label{eq:mixingpolicy}
    \mu_{pert}:=(1-\delta) \mu_{opt} + \delta \mu_{irr}, \quad 0<\delta < 1, 
\end{equation}
where $\delta$ is the perturbation degree.  
Clearly, this perturbed  policy $\mu_{pert}$ has the following two properties: 
\begin{itemize}[leftmargin=1.5em]
  \item
  First, we have $J^{\mu_{pert}}(s_0)\leq J^{\mu_{opt}}(s_0)$ 
  as $\mu_{opt}$ is already the optimal one to achieve the ratio objective. 
  Furthermore, $ J^{\mu_{pert}}(s_0)\to J^{\mu_{opt}}(s_0)$ as $\delta \to 0$; 
  \item 
  Second, all states in MDP will be visited infinitely often w.p.1. This is because, under policy $\mu_{pert}$, the system always has non-zero probability to execute irreducible policy $\mu_{irr}$. Furthermore, since $\M$ is an MAEC, it can finish LTL task w.p.1 under $\mu_{pert}$.
\end{itemize}

Now let us discuss how to use Proposition~\ref{prop:ratiodeviationequation} to determine the perturbation  degree $\delta$ such that $\epsilon$-optimality holds. 
Note that, in Equation~\eqref{eq:new}, 
term 
$ \textbf{D}_{\mathtt{R}}(\mu,\mu') -J^{\mu}(s_0, \mathtt{R}, \mathtt{C})\textbf{D}_{\mathtt{C}}(\mu,\mu')$ can be computed explicitly based on $\mu$ and $\mu'$. 
However, 
term $\frac{\pi_{\mu_{\delta}}^{\top}}{  \pi_{\mu_{\delta}}^\top v_\mathtt{C}^{\mu_{\delta}}     }    $ cannot be directly computed.
Our approach here is to estimate its bound as follows: 
\begin{itemize}[leftmargin=1.5em]
\item  
Let $c_{min} = \min_{s \in S, a\in A(s)} \mathtt{C}(s,a)$ be minimum cost for all state-action pairs.
Then we have $\pi_{\mu_{\delta}}^\top v_\mathtt{C}^{\mu_{\delta}} \geq c_{min} $.
\item 
Let the infinity norm of the computable part be
\begin{equation} \label{eq:devbound}
\textbf{D}_{\infty}^{\mu,\mu'} =\Vert  \textbf{D}_{\mathtt{R}}(\mu,\mu') -J^{\mu}(s_0)\textbf{D}_{\mathtt{C}}(\mu,\mu') 
\Vert_\infty.
\end{equation}
We have
$ | \pi_{\mu_{\zeta}}^{\top}(\textbf{D}_{\mathtt{R}}(\mu,\mu') -J^{\mu}(s_0)\textbf{D}_{\mathtt{C}}(\mu,\mu')) |\leq \textbf{D}_\infty^{\mu,\mu'}$.
\end{itemize}
These inequalities lead to the following result.

\begin{mypro} \label{prop:boundofrandom}
Let $\M=(S,s_0,A,P, \mathcal{AP},\ell,Acc)$ be a communicating MDP, 
$\mu_{opt} \in \Pi^S_\M$ be the optimal policy for ratio objective, 
$\mu_{irr}\in \Pi^S_\M$ be an irreducible policy, 
and $\mu_{pert}$ be defined in~\eqref{eq:mixingpolicy}. 
If
\begin{equation} \label{eq:deltabound}
       0< \delta \leq \epsilon  \frac{  c_{min}}{\textbf{D}_{\infty}^{\mu_{opt},\mu_{irr}}},
\end{equation}
then we have $J^{\mu_{pert}}(s_0) \geq J^{\mu_{opt}}(s_0) - \epsilon$. Furthermore, if $\M \in \texttt{MAEC}(\M)$, i.e., $\M$ itself is an MAEC, then $\mu_{pert}$ is a solution of Problem~\ref{problem:ratioltl} for $\M$.
\end{mypro}
\begin{proof}
To show \eqref{eq:deltabound}, we have
\begin{equation}\label{eq:choose}
    \begin{aligned}
    & \left | J^{\mu_{pert}}(s_0)-J^{\mu_{opt}}(s_0) \right | \\
    = & \frac{\delta \left | \pi_{\mu_{pert}}^{\top}(\textbf{D}_{\mathtt{R}}(\mu_{opt},\mu_{irr}) -J^{\mu}(s_0)\textbf{D}_{\mathtt{C}}(\mu_{opt},\mu_{irr})) \right |}{\pi_{\mu_{pert}}^{\top} v_\mathtt{C}^{\mu_{pert}}}  \\ 
        \leq &  \frac{\delta}{c_{min}} \left | \pi_{\mu_{pert}}^{\top}(\textbf{D}_{\mathtt{R}}(\mu_{opt},\mu_{irr}) -J^{\mu}(s_0)\textbf{D}_{\mathtt{C}}(\mu_{opt},\mu_{irr})) \right |  \\
         \leq & \frac{\delta}{c_{min}}  \pi_{\mu_{pert}}^{\top} \mathbf{1} \Vert\textbf{D}_{\mathtt{R}}(\mu_{opt},\mu_{irr}) -J^{\mu}(s_0)\textbf{D}_{\mathtt{C}}(\mu_{opt},\mu_{irr}) \Vert_{\infty}   \\
         = &\frac{\delta}{c_{min}} \textbf{D}_{\infty}^{\mu_{opt},\mu_{irr}}  \leq  \epsilon
         \nonumber
    \end{aligned}
\end{equation}
with $\mathbf{1} \in \mathbb{R}^{|S|}$ the vector where all elements are one. 
The first equality comes from Proposition~\ref{prop:ratiodeviationequation}. 
The first inequality holds since $\pi_{\mu_{pert}}^{\top}v_\mathtt{C}^{\mu_{pert}} \geq c_{min} \pi_{\mu_{pert}}^{\top} \mathbf{1} =c_{min}>0$.

Under policy $\mu_{pert}$, all states in $S$ will be visit infinitely often w.p.1. If $\M\in \texttt{MAEC}(\M)$, from definition of MAEC, we know that $\mu_{pert} \in \Pi^\varphi_\M$. From \eqref{eq:deltabound}, $\mu_{pert}$ is also $\epsilon$-optimal policy. Thus $\mu_{pert}$ is a solution of Problem~\ref{problem:ratioltl} for $\M$.
\end{proof}

\begin{remark}
    In general, we should perturb the optimal efficiency policy $\mu_{opt}$ by $\mu_{irr}$ to guarantee the satisfaction of LTL task. However, in some situation, the efficiency maximization may not conflict with LTL task, i.e., there exists $\mu_{opt} \in \Pi^S_\M \cap \Pi^\varphi_\M$ such that $J^{\mu_{opt}}(s_0)=J(s_0,\Pi^\varphi_\M)$. Then we can adopt $\mu_{opt}$ directly without perturbation. One can use procedure in \cite{chatterjee2015measuring} to check whether such stationary policy exists. If not, then we should perturb $\mu_{opt}$ by \eqref{eq:deltabound}.
\end{remark}

\begin{remark} \label{remark:irreduciblepolicy}
In this work, we select an irreducible policy $\mu_{irr}$ to perturb policy $\mu_{opt}$. However, here comes two issues: First, from analysis in this subsection, in general, we only need to select a stationary policy that can finish LTL task w.p.1 and ensure perturbed policy to induce unichain MC rather than an irreducible policy.
 Second, in our approach, we do not specify how to choose the irreducible policy $\mu_{irr}\in \Pi^S_\M$ and directly adopt the uniform policy. How to select a ``good'' policy to perturb $\mu_{opt}$ is beyond the scope of this paper and we take it as a future work. Thus we restrict on irreducible policy in this work for the sake of simplicity in expression.
However, if the efficiency of selected irreducible policy $\mu_{irr}$ is very small, 
then $\textbf{D}_{\infty}^{\mu_{opt},\mu_{irr}}$ will be very large. 
According to Equation~\eqref{eq:deltabound}, it means that we need to select a small perturbation degree $\delta$ 
to ensure $\epsilon$-optimality. 
Then, this also means that we will visit accepting states less frequently although they are still guaranteed to be visited infinitely often w.p.1, which may be undesirable when we want the interval between two arrivals of accepting states not too long. A direct heuristic approach is to obtain $\mu_{irr}$ by modifying  $\mu_{opt}$ so that their difference in efficiency is ``minimized". 
 \end{remark}

\begin{algorithm}[htp]
\caption{Solution for Communicating MDP}\label{alg:cmdpsolution} 
\KwIn{Threshold value $\epsilon > 0$ and communicating MDP $\M= (S,s_0,A,P,\mathcal{AP},\ell,Acc)$}
\KwOut{Optimal policy $\mu^\star \in \Pi^{S}_{\M}$ and its associated maximum efficiency $v^\star$ of MC $\M^{\mu^\star}$} 

Compute $\texttt{MAEC}(\M)=\{ (\S_1,\A_1),\dots,(\S_n,\A_n) \}$\\

\For{\text{each sub-MDP}  $(\S_i,\A_i),i=1,\dots,n$}
{
Compute the optimal objective value $v_i$ of program \eqref{opt1:obj}-\eqref{opt1:con5} and maximum efficiency policy $\mu^i_{opt}$ by Equation~\eqref{eq:optcmdppolicy} for MDP $(\S_i,\A_i)$
}
$i^\star\gets \arg\max_{i}\{ v_1,v_2,\dots,v_n\}$, $v^\star \gets v_{i^\star}$ \\

Compute irreducible policy $\mu_{irr}^i$ for $(\S_{i^\star},\A_{i^\star})$ \\

Pick $\delta>0$ satisfying \eqref{eq:deltabound} w.r.t. $\epsilon, \mu^{i^\star}_{opt}$ and $\mu_{irr}^{i^\star}$ \\

Get perturbed policy $\mu^{i^\star}_{pert}=(1-\delta)\mu^{i^\star}_{opt} + \delta \mu_{irr}^{i^\star}$ \\

For $s \in \S_{i^\star}$, $a \in \A_{i^\star}(s)$,  $\mu^\star(s,a)\gets \mu^{i^\star}_{pert}(s,a)$  \\
 $T\gets S\setminus \S_{i^\star}$, and  $G\gets \S_{i^\star} $ \\
\While{$T \neq \emptyset$}
{
Pick $s \in T,a \in A(s)$ \text{ s.t. } $\sum_{t \in G} P_{s,a,t}>0$ \\
 $\mu^\star(s,a)\gets 1$ \\
$T \gets T\setminus \{ s\}$, and $G\gets G \cup \{s \}$
}
\end{algorithm}

\subsection{Synthesis Algorithm for Communicating MDP} \label{sec:coverall}
Finally, Algorithm~\ref{alg:cmdpsolution} is proposed to solve Problem~\ref{problem:ratioltl} for any communicating MDP.
Specifically, we first compute all MAECs in MDP $\M$ and find the MAEC $(\S,\A)$ achieving highest efficiency value among all MAECs in lines 1-5. Then we compute an $\epsilon$-optimal perturbed policy over $(\S,\A)$ by result of Proposition~\ref{prop:boundofrandom} in lines 6-9 and ensure that all states in MDP will reach $\S$ eventually w.p.1 in lines 10-15.
 \begin{mythm} \label{thm:communicatingMDPprocedure}
 Let $\M= (S,s_0,A,P,\mathcal{AP},\ell,Acc)$ be a communicating MDP. Let $\mu^\star$ and $v^\star$ be the output policy and value of Algorithm~\ref{alg:cmdpsolution} when $\M$ and $\epsilon>0$ are input, respectively. Then $\mu^\star$ is a solution to Problem~\ref{problem:ratioltl} of $\M$ and $\forall s \in S, v^\star = J(s,\Pi_\M^\varphi)$.
 \end{mythm}
\begin{proof}
In line 5 we select the $(\S_{i^\star},\A_{i^\star}) \in \texttt{MAEC}(\M)$ achieving highest efficiency value among all MAECs. By perturbation in line 8, from result of Proposition~\ref{prop:boundofrandom}, $\S_{i^\star}$ consists a recurrent class in MC $\M^{\mu^\star}$. By action assignment procedure in lines 10-15, from \cite{puterman}, we know that $\M^{\mu^\star}$ has only one recurrent class $\S_{i^\star}$. Thus $\mu^\star \in \Pi^\varphi_\M$.

Consider any $\mu \in \Pi^\varphi_\M \cap \Pi^S_\M$. Let $R_1,R_2,\dots,R_K \subseteq S$ be the recurrent classes in $\M^{\mu}$. 
Since $\mu \in \Pi^{\varphi}_{\M}$, for any $R_k$, we can find $(\S,\A) \in \texttt{MAEC}(\M)$ such that $R_k \subseteq \S$. 
Let $r_k$ be the efficiency value restricted on recurrent class $R_k$ and $r_k^\star$ be maximum efficiency of the MAEC that $R_k$ belongs to. Then
\begin{equation}\label{eq:middleoffirsttheo}
J^\mu(s_0)=\sum_{k=1}^K \beta(k) r_k \leq \sum_{k=1}^K \beta(k) r^\star_k \leq v_{i^\star} \leq J^{\mu^\star}(s_0)+\epsilon,
\end{equation}
where $\beta(k)$ is the probability of staying forever in $R_k$ under policy $\mu$ such that $\sum_{k=1}^K\beta(k)=1$. The first equality comes from \eqref{eq:multichainratio}. The second inequality is right since line 5 of Algorithm~\ref{alg:cmdpsolution}. The third inequality holds from result of Proposition~\ref{prop:boundofrandom} and lines 7-8 of Algorithm~\ref{alg:cmdpsolution}. Then
\begin{equation}\label{eq:thmcmdpmiddle}
    J^{\mu^\star}(s_0)+\epsilon \geq J(s_0,\Pi^\varphi_\M \cap \Pi^{S}_\M) = J(s_0,\Pi^\varphi_\M)
\end{equation}
where last equality comes from Proposition~\ref{prop:stationaryisenough}. Thus $\mu^\star$ is a solution of Problem~\ref{problem:ratioltl} for $\M$. 

Since $\epsilon>0$ in \eqref{eq:middleoffirsttheo} can be arbitrary small in general, from third inequality of \eqref{eq:middleoffirsttheo} and $\mu^\star \in \Pi^{\varphi}_\M$, we know that $v^\star=J(s_0,\Pi^\varphi_\M)$. Since MC $\M^\star$ is a unichain, it holds that $J^{\mu^\star}(s)=J^{\mu^\star}(s')$ for any $s,s' \in S$. Then from \eqref{eq:thmcmdpmiddle} it holds that for any $s\in S$, $v^\star=J(s,\Pi^\varphi_\M)$.
\end{proof}

\begin{remark} \label{remark:applytonoinitial}
In Proposition~\ref{prop:boundofrandom}, we consider an MDP with initial state $s_0$. The initial state $s_0$ only plays a role in \eqref{eq:deltabound} since computation of $\textbf{D}_{\infty}^{\mu_{opt},\mu_{irr}}$ in \eqref{eq:devbound} requires value $J^{\mu_{opt}}(s_0)$. In Section~\ref{sec:review}, we know that MC $\M^{\mu_{opt}}$ is an unichain, which means that $\forall s,s' \in S,  J^{\mu_{opt}}(s)=J^{\mu_{opt}}(s')$. 
Thus the operation of computing $\delta$ in line 7 of Algorithm~\ref{alg:cmdpsolution} is well-defined although initial state of AMEC $(\S_{i^\star},\A_{i^\star})$ is not assigned. 
Moreover, from Theorem~\ref{thm:communicatingMDPprocedure} we know that the output value $v^\star$ of Algorithm~\ref{alg:cmdpsolution} is independent with initial state of the MDP. 
Thus we can still apply Algorithm~\ref{alg:cmdpsolution} to communicating MDP without knowing its initial state.
\end{remark}

\section{Solution to the General Case} \label{sec:general}
\subsection{Overview of Our Approach}
The approach in the previous section assumes that MDP $\M$ is communicating. 
In general, however, the MDP may not be communicating and 
the optimal ratio objective policy may induce a multi-chain MC, i.e., an MC containing more than one recurrent classes. 
Our approach for handling the general case consists of the following  steps: 
\begin{enumerate}[leftmargin=1.5em]
    \item 
    First, we decompose the MDP into several AMECs, i.e., communicating sub-MDPs in $\texttt{MEC}_\varphi(\M)$. 
    Eventually, the system needs to stay within AMECs in order to achieve the LTL task; 
    \item  
    Next, for each AMEC in $\texttt{MEC}_\varphi(\M)$, since it is communicating, 
    we can get solution of Problem~\ref{problem:ratioltl} and optimal efficiency value under LTL task constraint for the AMEC by inputting it to Algorithm~\ref{alg:cmdpsolution} in Section~\ref{sec:coverall}; 
    \item 
    Then, we construct a standard long-run average reward (per-stage) optimization problem, in which the reward for each state in AMEC is the optimal efficiency value under LTL task constraint of its associated AMEC. 
    Note that, since we consider long-run objective, the efficiency value only depends on the AMECs that it stays in forever.
    Therefore, the optimal policy of the average reward problem is a \emph{basic policy} determining which AMECs it should stay in forever. 
    \item 
    Finally, 
    for AMEC $(\S,\A) \in \texttt{MEC}_\varphi(\M)$, if it is recurrent under basic policy, it should be stayed in forever. Then we substitute basic policy by output policy of Algorithm~\ref{alg:cmdpsolution} over $(\S,\A)$ to get a final policy, which ensures $\epsilon$-optimality of the efficiency value and LTL task satisfaction.
\end{enumerate}

Before presenting our formal algorithm, 
we further introduce some necessary concepts.
Now suppose that $\M$ has $n$ AMECs, i.e., $\texttt{MEC}_\varphi(\M)=\{ (\S_1,\A_1),\dots,(\S_n,\A_n) \}$. 
For each AMEC $(\S_i,\A_i)$, 
we denote by $v_{i}^\star$ the output value of Algorithm~\ref{alg:cmdpsolution}, when $(\S_i,\A_i)$ is input.
Let $K\in \mathbb{R}$ be a real number. 
Then based on $K$ and $v^\star_i,i=1,\dots,n$, 
we define a new reward function    $\texttt{R}_{K}:S\times A \to \mathbb{R}$ for the entire $\M$ by: 
  \begin{align}\label{eq:rewardK}
	\texttt{R}_{K}(s,a) \!= \!
		\left\{\!\!
		\begin{array}{cl}
			v_i^\star  & \text{if }    s \in \S_i\wedge a \in \mathcal{A}_i(s)  \\
			K               & \text{otherwise}  
		\end{array}
		\right.\!\!\!.
\end{align}
Intuitively, for each state-action pair in an AMEC, 
the above construction assigns the reward identical to the optimal efficiency value under LTL task constraint one can achieve within this AMEC. 
For the remaining state-action pairs that are not in AMECs, 
we assign them value $K$. 
Clearly, 
to fulfill the LTL task, 
one needs to avoid executing state-action pairs with value $K$. 
Hence, the selected $K$ should be sufficiently small and we discuss it later in Section~\ref{sec:analysis}. 

Later on, we need to solve the classical long-run average reward maximization problem of $\M$
w.r.t. reward function $\texttt{R}_{K}$. 
We denote by $\mu^\star_K \in \Pi_\M^S$ the optimal long-run average reward policy, i.e., 
\begin{equation}\label{eq:optimalrewardKdef}
    W^{\mu^\star_K}(s,\texttt{R}_{K})= W(s,\texttt{R}_{K},\Pi_{\M}), \quad \forall s\in S.
\end{equation}
Such optimal policy $\mu^\star_K$ can be obtained by the standard linear programming approach in~\cite{puterman}, which can also be found in Appendix~\ref{app:program}.
 
\begin{algorithm}[tp]
\caption{Policy Synthesis for the General Case}\label{alg:gmdpsolution} 
\KwIn{MDP $\M= (S,s_0,A,P,\mathcal{AP},\ell,Acc)$ and threshold value $\epsilon>0$}
\KwOut{Policy $\mu^{\star} \in \Pi^{S}_{\M}$ which solve Problem~\ref{problem:ratioltl}} 

Compute $\texttt{MEC}_\varphi(\M)=\{ (\S_1,\A_1),\dots, (\S_n,\A_n) \}$\; 
Compute $\mu_i^\star$ and $v^\star_i$ for each $(\S_i,\A_i) \in \texttt{MEC}_{\varphi}(\M)$\; 
Define reward $\mathtt{R}_K$ according to Eq.~\eqref{eq:rewardK} and \eqref{eq:Kselect}\;  
Compute policy $\mu_K^\star$ by solving the classical long-run average reward maximization problem w.r.t.\ $\mathtt{R}_K$\; 
$\mu^\star\gets \mu_K^\star$\; 
\For{$(\S_i,\A_i) \in \texttt{AMEC}(\M)$}
{
\If{$\S_i$ contains a recurrent state in   MC $\M^{\mu_K^{\star}}$}
{
$\mu^\star(s,a)\gets0, \quad \forall s \in \S_i, a \in A(s)$ \\
$\mu^\star(s,a)\gets\mu^\star_i(s,a), \quad \forall s \in \S_i, a \in \A_i(s)$
}
}
\textbf{Return} 
$\epsilon$-optimal  policy $\mu^\star$
\end{algorithm} 

\subsection{Main Synthesis Algorithm}
Based on the above informal discussions, our  overall synthesis procedure for the entire MDP $\M$ is provided in Algorithm~\ref{alg:gmdpsolution}. 
Specifically, 
in line 1, we first compute AMECs set $\texttt{MEC}_\varphi(\M)$. 
Then we input each AMEC into Algorithm~\ref{alg:cmdpsolution} and record the output policy and value in line 2. 
These values help us to define reward function $\mathtt{R}_K$,
for which the maximum average reward policy $\mu^\star_K$ is synthesized. 
These are done by lines 3-4. 
Note that $K$ should satisfy \eqref{eq:Kselect} so that MDP will stay in AMEC states w.p.1.
Then in line~5, we choose $\mu^\star_K$ as the initial policy. 
Finally, in lines 6-11, 
we determine whether each AMEC $(\S_i,\A_i)$ contains some recurrent state in MC $\M^{\mu_K^{\star}}$. 
If so, it means that the MDP will achieve higher efficiency value when choosing to stay in this AMEC  forever. 
Therefore, within this AMEC, we replace $\mu^\star_K$ by output policy of Algorithm~\ref{alg:cmdpsolution} when this AMEC is input.
Note that, since each output policy is $\epsilon$-optimal within the AMEC by Theorem~\ref{thm:communicatingMDPprocedure}, the overall policy $\mu^\star$ is also $\epsilon$-optimal.

\subsection{Properties Analysis and Correctness}\label{sec:analysis}
We conclude this section by formally analyzing the properties of the proposed algorithm. 

The following result shows that, 
by selecting $K$ properly, 
the solution  to the long-run average reward maximization problem w.r.t. reward function $\mathtt{R}_K$
indeed achieves the supremum efficiency value among all policies in $\Pi_\M^\varphi$. 

\begin{mypro} \label{prop:rewardisequaltoratio}
Let $\hat{r}=\max_{s\in S, a \in A(s)}|\mathtt{R}(s,a)|$ and $\hat{c}=\min_{s\in S, a \in A(s)}\mathtt{C}(s,a)$. If $K$ is selected such that
\begin{equation}\label{eq:Kselect}
    K< -\frac{\hat{r}}{\hat{c}},
\end{equation}
then we have 
\begin{equation}\label{eq:rewardKequaltaskratio}
    W(s_0,\mathtt{R}_K,\Pi_\M) = J(s_0,\mathtt{R},\mathtt{C},\Pi^\varphi_\M).
\end{equation}
 \end{mypro}
\begin{proof}
    Let $\mu^\star_K$ the optimal stationary policy for average reward w.r.t. $\mathtt{R}_K$, i.e.,
    \begin{equation} \label{eq:averageoptimal}
        W^{\mu^\star_K}(s,\mathtt{R}_K)=W(s,\mathtt{R}_K,\Pi_\M), \forall s\in S.
    \end{equation}
Existence of such policy $\mu^\star_K$ comes from classic average maximization problem~\cite{puterman}.
We first prove that for $K < -\hat{r}/\hat{c}$, all recurrent states of $\M^{\mu_K^\star}$ are states in AMECs by contradiction. Assume that $r$ is recurrent in MC $\M^{\mu^\star_K}$ and is not in any AMEC. Let $R$ the recurrent class $r$ belongs to in $\M^{\mu^\star_K}$.
Then
\begin{equation} \label{eq:smallKprove}
 W^{\mu^\star_K}(r,\mathtt{R}_K) = K.
\end{equation}
Since we assume that initial from $r$ it can finish LTL w.p.1 under some policy $\mu'$, then $r$ can stay in forever in AMECs w.p.1 under $\mu'$. From claim~\ref{claim:boundforLTLstate} and \eqref{eq:ratiomiddle}, we know that $W^{\mu'}(s,\mathtt{R}_K) \geq -\hat{r}/\hat{c}>K$, which violates \eqref{eq:averageoptimal}. Thus all recurrent states in MC $\M^{\mu_K^\star}$ are in AMECs. 
From Claim~\ref{cla:regularoptimalpolicy} in Appendix~\ref{appendix:proof}, we know that $\mu^\star_K$ is regular.
Let $\texttt{R}(\M) \subseteq \texttt{MEC}(\M)$ be the set of MECs that contain recurrent states in MC $\M^{\mu^\star_K}$.
For $(\S,\A)$, we denote by $\mu_{(\S,\A)}$ the output policy of Algorithm~\ref{alg:cmdpsolution} when $(\S,\A)$ and $\epsilon > 0$ is input. We define a policy $\hat{\mu}$ by
\begin{align}
	\hat{\mu}(s,a) \!= \!
		\left\{\!\!
		\begin{array}{cl}
			\mu_{(\S,\A)}(s,a)  & \text{if }    s\in \S, a \in \A(s), (\S,\A) \in \texttt{R}(\M)  \\
			\mu^{\star}(s,a)               & \text{otherwise.}  
		\end{array}
		\right.\!\!\! \nonumber
  \end{align}
Then we have
\begin{equation} \label{eq:mukstarreward}
\begin{aligned}
     &W^{\mu^\star_K}(s_0,\mathtt{R}_K) \\
     \overset{(a)}{=}& \sum_{(\S,\A) \in \texttt{AMEC}(\M)} \textsf{Pr}^{\mu^\star_K}_\texttt{R}(\S,\A) \mathtt{R}_K(\S,\A) \\
     \overset{(b)}{=} & \sum_{(\S,\A) \in \texttt{AMEC}(\M)}\textsf{Pr}^{\hat{\mu}}_\texttt{R}(\S,\A) V(\S,\A)\\
     \overset{(c)}{\leq} & \sum_{(\S,\A) \in \texttt{AMEC}(\M)}\textsf{Pr}^{\hat{\mu}}_\texttt{R}(\S,\A) (J^{\hat{\mu}}(s_{(\S,\A)},\mathtt{R},\mathtt{C})+\epsilon) \\ 
    \overset{(d)}{=} & J^{\hat{\mu}}(s_0,\mathtt{R},\mathtt{C}) +\epsilon \\
     \overset{(e)}{\leq} &J(s_0,\mathtt{R},\mathtt{C},\Pi_\M^\varphi)+\epsilon.
\end{aligned}
\end{equation}
where $\textsf{Pr}^{\mu^\star_K}_\texttt{R}(\S,\A)$ and $\textsf{Pr}^{\hat{\mu}}_\texttt{R}(\S,\A)$ defined in \eqref{eq:stayingforeverinMEC} are probability of staying forever in MEC $(\S,\A)$ under policy $\mu^\star_K$ and $\hat{\mu}$, respectively, and $\mathtt{R}_K(\S,\A)$ is the constant reward assigned for state action pairs in $(\S,\A)$ by function $\mathtt{R}_K$ in \eqref{eq:rewardK}, and $V(\S,\A)$ is output value of Algorithm~\ref{alg:cmdpsolution} when $(\S,\A)$ and $\epsilon$ is input, with $s_{(\S,\A)} \in \S$ some state in $\S$. $(a)$ holds from \eqref{eq:MECreachingratio}. $(b)$ comes from \eqref{eq:samestayingprobability} and definition of $\mathtt{R}_K$ in \eqref{eq:rewardK}. $(c)$ holds from result of Theorem~\ref{thm:communicatingMDPprocedure} and $(d)$ comes from \eqref{eq:MECreachingratio}. $(e)$ is true since $\hat{\mu} \in \Pi^\varphi_\M$.
Since $\epsilon>0$ can be selected arbitrarily closed to $0$, we have
\begin{equation}\label{eq:anotherhalf}
    W^{\mu^\star_K}(s_0,\mathtt{R}_K) \leq J(s_0,\mathtt{R},\mathtt{C},\Pi_\M^\varphi).
\end{equation}
  Let $\mu^\star \in \Pi^{SD}_\M$ be the stationary deterministic policy defined in Eq.~\eqref{eq:optimaldeterdef}. Then from 1) of Claim~\ref{claim:deterpolicyproperty} in Appendix~\ref{appendix:proof} and \eqref{eq:MECreachingratio}, we have
  \begin{equation}\label{eq:middlerewardK}
         W^{\mu^\star}(s_0,\mathtt{R}_K)=\sum_{(\S,\A) \in \texttt{AMEC}(\M)} \textsf{Pr}^{\mu^\star}_\texttt{R}(\S,\A) \mathtt{R}_K(\S,\A)
  \end{equation}
where $\mathtt{R}_K(\S,\A)$ is defined in \eqref{eq:mukstarreward} and $\textsf{Pr}^{\mu^\star}_\texttt{R}(\S,\A)$ is defined in \eqref{eq:stayingforeverinMEC}. From 1) of Claim~\ref{claim:deterpolicyproperty}, each recurrent class of MC $\M^{\mu^\star}$ is in some MAEC. From Theorem~\ref{thm:communicatingMDPprocedure} and definition of $\mathtt{R}_K$ in \eqref{eq:rewardK}, we know that $\mathtt{R}_K(\S,\A)$ is the value of maximum efficiency among all MAECs in $(\S,\A)$, i.e.,
\begin{equation}\label{eq:Krewardupper}
    \mathtt{R}_K(\S,\A) \geq J^{\mu^\star}(s_{(\S,\A)},\mathtt{R},\mathtt{C}),
\end{equation}
 where $s_{(\S,\A)} \in \S$ is defined in \eqref{eq:mukstarreward}. From \eqref{eq:middlerewardK},\eqref{eq:Krewardupper} and \eqref{eq:MECreachingratio},
\[
W^{\mu^\star}(s_0,\mathtt{R}_K) \geq J^{\mu^\star}(s_0,\mathtt{R},\mathtt{C}).
\]
Combining with 2) of Claim~\ref{claim:deterpolicyproperty}, \eqref{eq:stationaryhalf}, \eqref{eq:optimaldeterdef}, we have\begin{equation}\label{eq:generalhalf}
    W(s_0,\mathtt{R}_{K}, \Pi_\M) \geq W^{\mu^\star}(s_0,\mathtt{R}_K) \geq J(s_0,\mathtt{R},\mathtt{C},\Pi^\varphi_{\M}).
\end{equation}
With \eqref{eq:optimalrewardKdef}, \eqref{eq:anotherhalf} and \eqref{eq:generalhalf}, we complete the proof.
\end{proof}

Based on the above criterion, we can finally establish the correctness result of the synthesis procedure for the general case of non-communicating MDPs.
 \begin{mythm} \label{thm:solutionofproblem}
Given MDP $\M= (S,s_0,A,P,\mathcal{AP},\ell,Acc)$.
Algorithm~\ref{alg:gmdpsolution} correctly solves Problem~\ref{problem:ratioltl} for $\M$.
 \end{mythm}
\begin{proof}
By Proposition~\ref{prop:rewardisequaltoratio}, we know that policy $\mu^\star$ in lines 5 of Algorithm~\ref{alg:gmdpsolution} satisfies that all recurrent states in MC $\M^{\mu^\star}$ is in AMECs. 
Then after action assignment procedure in lines 6-11, we get a policy $\mu^\star$ with efficiency value
\begin{equation} \label{eq:outputofgmdpalg}
    J^{\mu^\star}(s_0,\mathtt{R},\mathtt{C})\geq \sum_{(\S,\A)\in \texttt{AMEC}(\M)}\textsf{Pr}^{\mu^\star}_\texttt{R}(\S,\A) (v^\star(\S,\A)-\epsilon)
\end{equation}
 where $\textsf{Pr}^{\mu^\star}_\texttt{R}(\S,\A)$ defined by \eqref{eq:stayingforeverinMEC} is probability of staying forever in $(\S,\A)$ and $v^\star(\S,\A)$ is the output of Algorithm~\ref{alg:cmdpsolution} when $(\S,\A)$ and threshold value $\epsilon$ are input. From definition of $\mathtt{R}_K$ in \eqref{eq:rewardK}, for $s \in \S, a \in \A(s)$, $v^\star(\S,\A)=\mathtt{R}_K(s,a)$.
 
 From Claim~\ref{cla:regularoptimalpolicy} we know $\mu^\star_K$ is regular. Then from \eqref{eq:samestayingprobability}, for any $(\S,\A) \in \texttt{AMEC}(\M)$, $\textsf{Pr}^{\mu^\star}_\texttt{R}(\S,\A)=\textsf{Pr}^{\mu^\star_K}_\texttt{R}(\S,\A)$. Combining with first equality of \eqref{eq:mukstarreward}, \eqref{eq:outputofgmdpalg}, \eqref{eq:rewardKequaltaskratio} and \eqref{eq:averageoptimal},
 \begin{equation} \label{eq:gmdpalgmiddle}
      J^{\mu^\star}(s_0,\mathtt{R},\mathtt{C})+\epsilon \geq W^{\mu_K^\star}(s_0,\mathtt{R}_K)=J(s_0,\mathtt{R},\mathtt{C},\Pi^\varphi_\M).
 \end{equation}
 Since under $\mu_i^\star$ in line 9 of Algorithm~\ref{alg:gmdpsolution}, it can finish LTL task w.p.1 once reaching AMEC $(\S_i,\A_i)$, we have $\mu^\star \in \Pi^\varphi_\M$. Then by \eqref{eq:gmdpalgmiddle} we know $\mu^\star$ is a solution of problem~\ref{problem:ratioltl} for $\M$.
\end{proof}

\begin{remark}
We briefly discuss the complexity of Algorithm~\ref{alg:gmdpsolution}, which arises from the following three main components: the computation of the MEC in line 1, the call to Algorithm~\ref{alg:cmdpsolution} in line 2, and the solution of the average reward maximization problem in line 4.
From~\cite{baier2008principles}, the complexity of line 1 is quadratic in the size of the MDP. Additionally, the average reward maximization problem involves solving a linear program, and its optimal solution can be computed in polynomial time with respect to the size of $\M$~\cite{boyd2004convex}. Finally, the complexity of Algorithm~\ref{alg:cmdpsolution} is also polynomial in the size of the MDP. Specifically, it requires finding the MAEC in line 1, solving a linear fractional program to compute the optimal ratio value in line 3, and performing a matrix inversion in the potential vector \eqref{eq:potentialvector} to select $\delta$ in line 7. 
Therefore, the overall complexity of our approach is polynomial in the size of the (product) MDP.
\end{remark}

\section{Case Studies} \label{sec:casestudy}
In this section, we present two case studies of robot task planning to illustrate the proposed method. 
All computations are performed on a laptop with 16 GB RAM.  We use \textsf{CVXPY}~\cite{diamond2016cvxpy} to solve convex optimization problems.

\subsection{Case Study 1}

\textbf{Mobility of Robot: }We consider a mobile robot moving in a $9 \times 9$ gird workspace shown in Figure~\ref{fig:workspace}. 
The initial location of the robot is the blue grid in the upper left corner 
and red girds represent obstacle regions the robot cannot enter. 
We assume that the mobility of the robot is fully deterministic. 
That is, at each gird, the robot has at most four actions, left/right/up/down, 
and the robot can deterministically move to the unique corresponding successor grid by taking each action. 
An action is not available if it leads to the boundary.
Therefore, the mobility of the robot can be modeled as a deterministic MDP denoted by 
$\hat{\M}=(\hat{S}, \hat{s}_0, \hat{P}, \hat{A})$ with state space $\hat{S}=\{(i,j):i,j=1,\dots,9\}$.

\textbf{Probabilistic Environment: }
We assume that, at each time instant, when the robot is at grid $\hat{s}\in \hat{S}$, it has probability $p(\hat{s})$ to find an item; 
the probability distribution over the workspace is shown in  Figure~\ref{fig:generation probability}.  
If the robot is empty, then it will pick up the item immediately when find it and the robot can only carry at most one item. The robot delivers the items to one of the destinations in $\hat{D} \subseteq S$, which are denoted by green grids.

\textbf{MDP Model: }The overall behavior of both the deterministic mobility and the probabilistic environment  can be captured by MDP $\M=({S}, {s}_0, {P}, A=\hat{A})$
with augmented state space $S=\hat{S}\times \{ 0,1 \}$, 
where $0$ means that the robot is empty and $1$ means that it is carrying item. 
We assume that the robot is initially empty, i.e.,  $s_0 = (\hat{s}_0,0)$. 
We denote by 
$D=\hat{D}\times \{1\}$ the set of states where the robot is at the destination with item.
Then the transition probability is defined by: 
for any states $s=(\hat{s},i)$, $s'=(\hat{s}',i')$ and action $a \in A=\hat{A}$, 
we have 
\begin{enumerate}
\item 
if $\hat{P}(\hat{s}'\mid \hat{s},a)=0$, then $P(s'\mid s,a)= 0$; 
\item 
otherwise, we have
\end{enumerate} 
\begin{align} 
	 P(s'|s,a)= 
		\left\{
		\begin{array}{ll}
			 p(\hat{s}')    &  \text{if }  i=0\wedge i'=1 \\
            1-p(\hat{s}')    &  \text{if }  i=0 \wedge i'=0 \\
		1          &  \text{if }  i=1 \wedge i'=1 \wedge s\notin D \\ 
			 p(\hat{s}')    &  \text{if }  i=1 \wedge i'=1 \wedge s\in D \\
            1-p(\hat{s}')    &  \text{if }  i=1 \wedge i'=0 \wedge s\in D\\
		\end{array}.
		\right. \nonumber
\end{align}

\textbf{LTL Task: } The states in $D$ are assigned with label $d$. The yellow grid in lower left of Figure~\ref{fig:workspace}, denoted by $c$, is charging station. Let $C=\{ c \}\times \{0,1\}$ be all states in augment MDP that represent charging station with label $c$. The obstacle regions have label $b$. We consider two LTL tasks in the case study. The first task is described by
\begin{equation}
  \varphi_1 = \square \lozenge \nonumber d \wedge \square \neg b,
\end{equation}
i.e., the robot need to find and pick up items in the workspace and then deliver
to the destinations while avoid the obstacles. The second task further take in energy constraint consideration such that the charging station should also be visited infinitely often. Thus we have
\begin{equation}
    \varphi_2 = \varphi_1 \wedge \square \lozenge c. \nonumber
\end{equation}

\textbf{Costs and Rewards: }We assume that moving from each grid incurs a cost. 
Specifically, for each state $s=(\hat{s},i)\in S$ and action $a\in A$, the moving cost $\mathtt{C}(s,a)$ is defined by 
$\mathtt{C}(s,a)=\textsf{cost}(M(\hat{s}))$, 
where $M(\hat{s})$ is the shortest Manhattan distance from $\hat{s}$ to target grids
and $\textsf{cost}(\cdot)$ is shown in Table~\ref{table}. The robot receives a reward when reaching the destinations with item. 
Assume that the rewards for destinations $s_{ll}\in D$  in the lower left   and $s_{ur}\in D$ in the upper right corner are different
with $\mathtt{R}(s_{ll},a)=2$ and $\mathtt{R}(s_{ur},a)=1$ for all $a\in A$. 
The overall objective of the robot is to finish LTL task w.p.1
while maximizing the expected reward-to-cost ratio.
 \begin{table}[tp]
  \begin{center}
    \caption{Cost function based on Manhattan distance}
    \label{table}
    \begin{tabular}{|c|c|c|c|c|c|c|c|c|c|} 
    \hline
     Distance  & 0 & 1 & 2 & 3  & 4 & 5 & 6 & 7  & 8\\
      \hline
      Cost  & 3.2 & 3.0 & 2.7 & 2.5 & 1.5 & 1.0 & 1.0 & 1.0  & 1.0 \\
      \hline
    \end{tabular}
  \end{center}
\end{table}
\begin{figure}  
\subfigure[Workspace of the robot, where arrows indicate optimal actions.] 	{\label{fig:workspace}
    \begin{minipage}[b]{0.47\linewidth}
   \centering
\begin{tikzpicture}
       \draw[fill=myblue, draw = white] (0,0) -- (0.4,0) -- (0.4, -0.4) -- (0, -0.4);

      \draw[fill=mygreen, draw = white] (3.2,0) -- (3.6,0) -- (3.6, -0.4) -- (3.2, -0.4);	
       \draw[fill=mygreen, draw = white] (0,-3.2) -- (0,-3.6) -- (0.4,-3.6) -- (0.4,-3.2);

       \draw[fill=red, draw = white]
       (0.4,-0.4)--(1.2,-0.4)--(1.2,-1.2)--(0.4,-1.2);

        \draw[fill=red, draw = white]
       (0.8,-2.4)--(1.6,-2.4)--(1.6,-3.2)--(0.8,-3.2);

        \draw[fill=red, draw = white]
       (2.4,-1.2)--(3.2,-1.2)--(3.2,-2.0)--(2.4,-2.0);

       \draw[fill=yellow, draw = white]
       (0,-2.8)--(0.4,-2.8)--(0.4,-3.2)--(0,-3.2);

        \foreach \x in {0, 0.4,...,3.6}
        \draw (\x, 0) -- (\x, -3.6);
	
	\foreach \y in {0, -0.4,...,-3.6}
	\draw[densely dotted] (0, \y)--(3.6, \y);

    \draw [line width =1.5pt] (0, 0)--(3.6, 0)--(3.6,-3.6)--(0,-3.6)--(0,0);
    
    	\begin{scope}[red, line width=0.2mm, arrows = {-Stealth[scale=0.5]}]

            \draw (0.75, -3.47) -- (0.45, -3.47);
            \draw (1.15, -3.47) -- (0.85, -3.47);
            \draw (1.55, -3.47) -- (1.25, -3.47);
            \draw (1.95, -3.47) -- (1.65, -3.47);
            \draw (2.35, -3.47) -- (2.05, -3.47);
    	\end{scope}

         	\begin{scope}[blue, line width=0.2mm, arrows = {-Stealth[scale=0.5]}]

            \draw (0.05,-3.4) -- (0.35,-3.4);

            \draw (0.45,-3.3) -- (0.75,-3.3);
            \draw (0.85,-3.3) -- (1.15,-3.3);
            \draw (1.25,-3.3) -- (1.55,-3.3);
            \draw (1.65,-3.3) -- (1.95,-3.3);
            \draw (2.35,-3.3) -- (2.05,-3.3);
    	\end{scope}
    
\end{tikzpicture}
           \end{minipage}
	}
\subfigure[Probabilities for finding items.] 
	{\label{fig:generation probability}
 \begin{minipage}[b]{0.45\linewidth}
	\centering
 \includegraphics[height=3.6cm]{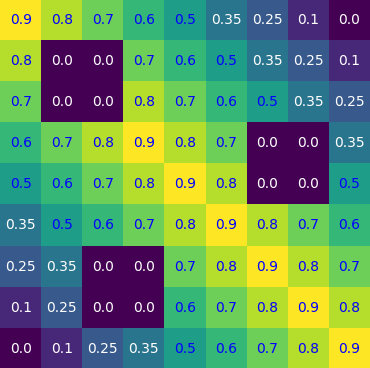}
    \end{minipage}
	}
    \caption{Case study 1.}
   \label{fig:case1}
\end{figure}

\begin{figure}[tp]
	\subfigure[Limit distribution for $\hat{S}\times \{1\}$.] 
	{\label{fig:limitdishave}
 \begin{minipage}[b]{0.47\linewidth}
	\centering
 \includegraphics[height=0.4cm]{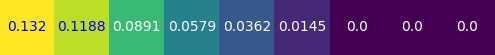}
    \end{minipage}
	}\subfigure[Limit distribution for $\hat{S}\times \{0\}$.] 
	{\label{fig:limitdisno}
 \begin{minipage}[b]{0.47\linewidth}
	\centering
 \includegraphics[height=0.4cm]{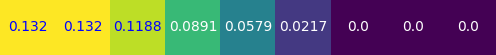}
    \end{minipage}
	}

    \caption{Limit distribution under the optimal policy.}
   \label{fig:limitdis}
\end{figure}

\textbf{Solution Analysis for Task $\varphi_1$: }  
By applying the synthesis algorithm,  the robot will first take an arbitrary transient path and then eventually  circulate along the path shown in Figure~\ref{fig:workspace}. 
Specifically, the red and blue arrows indicate the action  robot should take if it has and has not picked up the items, respectively.
The optimal efficiency value computed is $0.1157$. 
The limit distribution under this policy is shown in Figure~\ref{fig:limitdis}, which only illustrates grids in last row since only these grids have non-zero value.   
For other girds, as stated in Equation~\eqref{eq:optcmdppolicy}, 
they are all transient states and the  optimal action is an action under which robot can reach these six grids. 
Since the synthesized policy can already finish the LTL tasks, according to Remark~4, we do not even need to perturb the policy. 
Note that, there are two considerations to form this solution.  
First, the robot can choose to go to destinations either $s_{ll}$ and $s_{ur}$. 
However, the former one gives more reward. 
Second, as shown in Figure~\ref{fig:generation probability}, 
the further away from the destination state, the greater the probability of finding the item, but also the higher the overall cost incurs.
Therefore, there is a trade-off to decide how far away the robot should leave from the destination, and one solution is actually the optimal one. 

\textbf{Solution Analysis for Task $\varphi_2$: }
Since the reward and cost function are same for LTL tasks $\varphi_1$ and $\varphi_2$, the policy illustrated in Figure~\ref{fig:workspace} also achieves highest ratio objective value. However, LTL task $\varphi_2$ requires robot to visit charge station, i,e., the yellow grid, infinitely often, which means that this optimal efficiency policy can not finish LTL task w.p.1. To ensure the qualitative requirement, we consider irreducible policy such that at grid $(8,1)$, the target grid is $(9,1)$ and at grid $(9,1)$, the target grid is $(8,1)$, i.e., the robot will repeat back and forth between two grids. For remaining states, we choose action under which it will reach these two grids in finite steps. Note that as discussed in Remark~\ref{remark:irreduciblepolicy}, the chosen irreducible policy not necessarily induces an irreducible MC. We only need to ensure perturbed policy to induce unichain MC and finish LTL task w.p.1. The ratio objective value is $0.032$ under irreducible policy. 

In order to select the perturbation degree (PD) $\delta$ to ensure $\epsilon$-optimality, we consider both estimation method (ES) and exact method (EX). Specifically, to select PD, ES method uses \eqref{eq:deltabound} in Proposition~\ref{prop:boundofrandom} and EX method uses bisection search over interval $(0,1)$. The result is illustrated in Table~\ref{tab:perturbationdegree}-\ref{tab:limitdistribution}.
Specifically, for each sub-optimal threshold (TH) $\epsilon$, we perturb optimal efficiency policy by irreducible policy using ES and EX methods and record two values: (1) the selected PD in Table~\ref{tab:perturbationdegree}; (2) the limit probability of charging station (yellow grid) in Table~\ref{tab:limitdistribution}.
From Table~\ref{tab:perturbationdegree}, the PD in \eqref{eq:deltabound} is about an order of magnitude smaller than maximum PD, which is the cost of determining PD faster. One drawback of low PD is that the limit probability of visiting charging station is also low, as shown in Table~\ref{tab:limitdistribution}.
Note that in Problem~\ref{problem:ratioltl}, we only require that the synthesized policy is $\epsilon$-optimal among all policies that can finish LTL task w.p.1. The LTL task $\varphi_2$ only requires charge station can be visited infinitely often, which is equivalent to non-zero limit probability of visiting charge station. Therefore, although by ES method the selection of PD is conservative and limit probability is relatively low, it still achieve requirement of Problem~\ref{problem:ratioltl}. However, in some situation, it is better to have higher frequency of visiting charging station under $\epsilon$-optimal constraint. Such requirement cannot be expressed by LTL formula and is out of scope of this work. We regard it as a future direction to investigate how to optimize the quantitative objective under constraint of visiting frequency in infinite horizon.

     \begin{table}[tpb]
      \caption{Perturbation Degree $\delta$ for Different Threshold}
  \label{tab:perturbationdegree}
  \centering
  \begin{threeparttable}[b]
     \begin{tabular}{cccccccccc}
      \toprule
     TH & $0.005$ & $0.01$  & $0.05$  & $0.1$  \\ 
      \midrule
     ES  & $0.0087$ & $0.017$ & $0.087$  & $0.17$ \\
     EX  & $0.064$ & $0.12$ & $0.48$  & $0.99$\\
      \bottomrule
    \end{tabular}
  \end{threeparttable}
\end{table}

     \begin{table}[tpb]
      \caption{Limit Probability of Charing Station for Different Threshold}
  \label{tab:limitdistribution}
  \centering
  \begin{threeparttable}[b]
     \begin{tabular}{cccccccccc}
      \toprule
     TH & $0.005$ & $0.01$  & $0.05$  & $0.1$  \\ 
      \midrule
     ES  & $0.0012$ & $0.0024$ & $0.013$  & $0.03$ \\
     EX  & $0.0093$ & $0.019$ & $0.14$  & $0.48$\\
      \bottomrule
    \end{tabular}

  \end{threeparttable}
\end{table}

\subsection{Case Study 2}
\textbf{Mobility of Robot: } The workspace of robot is a $7\times 7$ smart factory shown in Figure~\ref{fig:workspace2}. The initial location of the robot is indicated by the black arrow. The mobility of the robot is deterministic like that in case study 1. In each grid, the robot has available action whose target grid is indicated by the blue and red arrows. The mobility of robot can be modeled as a deterministic MDP.

\textbf{LTL Task: }The red grid, labeled by $r$, is material station where robot can get spare part.
The green grid, labeled by $g$, is command center where the robot can get permission to obtain spare part. The robot is required to first reach command center to get permission and then reach material station to obtain spare part infinitely often. The LTL task is described by
\[
\varphi = \square (\lozenge ( g \wedge \lozenge r ) ).
\]

\textbf{Costs and Rewards:} Once reaching each grid, the robot will receive a reward and cost, which is the red and green number over corresponding gird in Figure~\ref{fig:case2reward}. Specifically, the reward represents the amount of spare part transporting to each grid and the cost represents the time consumption moving to each grid. The quantitative objective of robot is to maximize the ratio of reward and cost, i.e., the transportation spare part amount per time instant.

\textbf{Solution Analysis: }We first find the optimal efficiency policy under which robot will execute the red arrows action infinitely often receiving efficiency $1.1$. However, this policy cannot finish the LTL task. We can perturb the policy by applying the synthesis algorithm as case study 1 to find a solution of Problem~\ref{problem:ratioltl}. One may ask whether it is possible to modify the reward function to encourage the robot to finish the LTL task such that the optimal efficiency policy can directly finish LTL task w.p.1 and the perturbation procedure may be unnecessary. To this end, we consider a new reward function $\mathtt{R}(i)$ which is same as origin reward function but adds reward $i$ when robot successfully obtains spare part in material station. If $i \leq 63.53$, the optimal efficiency policy is still the same as that under origin reward function. For $i > 65.53$, the optimal efficiency policy is replaced by a new policy under which the robot will repeatedly first go to green grid and then go to red grid. Under such policy, the transportation spare part amount per time instant is $0.75$. By considering the reward functions $\{ \mathtt{R}(i) \}_{i\geq 0}$, there are only two different optimal efficiency policies. Although we can find a optimal efficiency policy satisfying LTL task w.p.1 by selecting sufficiently large $i$, the robot may get a undesired actual efficiency ($0.75$ in this case). Therefore, we may get undesired policy by trivial reward modification. 

\begin{figure}  
\subfigure[Workspace of case 2] 
	{\label{fig:workspace2}
            \begin{minipage}[b]{0.46\linewidth}
               	\centering
\begin{tikzpicture}
\draw[fill=mygreen, draw = white]
       (1.5,-1)--(2,-1)--(2,-1.5)--(1.5,-1.5);

\draw[fill=red, draw = white]
      (3,-3)--(3.5,-3)--(3.5,-3.5)--(3,-3.5);

\draw[fill=black, draw = white]
(1.5,-1.5)--(2,-1.5)--(2,-2)--(1.5,-2);

	\foreach \x in {0,0.5,...,3.5}
	\draw[densely dotted,draw=black] (\x, 0)--(\x, -3.5);
		
	\foreach \y in {0,-0.5,...,-3.5}
	\draw[densely dotted,draw=black] (0,\y)--(3.5,\y);
		
	\draw [line width =1.5 pt,draw=black] (0,0)--(3.5, 0)--(3.5, -3.5)--(0,-3.5)--(0,0);

 \draw [line width =1.5 pt,draw=black] (0.5, -1.5)--(0.5, -0.5)--(3,-0.5)--(3,-3)--(0.5,-3)--(0.5,-2);

  \draw [line width =1.5 pt,draw=black] (1, -1.5)--(1, -1)--(2.5,-1)--(2.5,-2.5)--(1,-2.5)--(1,-2);

        \begin{scope}[blue, line width=0.2mm, arrows = {-Latex[length=1.2mm]}]
    	\draw (0.25,-1.65) -- (0.25,-1.35);
     \draw (0.25,-1.15) -- (0.25,-0.85);
      \draw (0.25,-0.65) -- (0.25,-0.35);
      
      \draw (0.35,-0.25) -- (0.65,-0.25);
      \draw (0.85,-0.25) -- (1.15,-0.25);
      \draw (1.35,-0.25) -- (1.65,-0.25);
      \draw (1.85,-0.25) -- (2.15,-0.25);
      \draw (2.35,-0.25) -- (2.65,-0.25);
      \draw (2.85,-0.25) -- (3.15,-0.25);

      \draw (3.25,-0.35) -- (3.25,-0.65);
      \draw (3.25,-0.85) -- (3.25,-1.15);
      \draw (3.25,-1.35) -- (3.25,-1.65);
      \draw (3.25,-1.85) -- (3.25,-2.15);
      \draw (3.25,-2.35) -- (3.25,-2.65);
      \draw (3.25,-2.85) -- (3.25,-3.15);

      \draw (3.15,-3.25) -- (2.85,-3.25);
      \draw (2.65,-3.25) -- (2.35,-3.25);
      \draw (2.15,-3.25) -- (1.85,-3.25);
      \draw (1.65,-3.25) -- (1.35,-3.25);
      \draw (1.15,-3.25) -- (0.85,-3.25);
      \draw (0.65,-3.25) -- (0.35,-3.25);

      \draw (0.25,-3.15) -- (0.25,-2.85);
      \draw (0.25,-2.65) -- (0.25,-2.35);
      \draw (0.25,-2.15) -- (0.25,-1.85);
    	\end{scope}    

             \begin{scope}[red, line width=0.2mm, arrows = {-Latex[length=1.2mm]}]
    	\draw (0.75,-1.65) -- (0.75,-1.35);
     \draw (0.75,-1.15) -- (0.75,-0.85);
      
      \draw (0.85,-0.75) -- (1.15,-0.75);
      \draw (1.35,-0.75) -- (1.65,-0.75);
      \draw (1.85,-0.75) -- (2.15,-0.75);
      \draw (2.35,-0.75) -- (2.65,-0.75);

      \draw (2.75,-0.85) -- (2.75,-1.15);
      \draw (2.75,-1.35) -- (2.75,-1.65);
      \draw (2.75,-1.85) -- (2.75,-2.15);
      \draw (2.75,-2.35) -- (2.75,-2.65);

      \draw (2.65,-2.75) -- (2.35,-2.75);
      \draw (2.15,-2.75) -- (1.85,-2.75);
      \draw (1.65,-2.75) -- (1.35,-2.75);
      \draw (1.15,-2.75) -- (0.85,-2.75);

      \draw (0.75,-2.65) -- (0.75,-2.35);
      \draw (0.75,-2.15) -- (0.75,-1.85);
    	\end{scope} 

    \begin{scope}[blue, line width=0.2mm, arrows = {-Latex[length=1.2mm]}]
    \draw (1.25,-1.65) -- (1.25,-1.35);

      \draw (1.35,-1.25) -- (1.65,-1.25);
      \draw (1.85,-1.25) -- (2.15,-1.25);

      \draw (2.25,-1.35) -- (2.25,-1.65);
      \draw (2.25,-1.85) -- (2.25,-2.15);

      \draw (2.15,-2.25) -- (1.85,-2.25);
      \draw (1.65,-2.25) -- (1.35,-2.25);

      \draw (1.25,-2.15) -- (1.25,-1.85);
    	\end{scope}

     \begin{scope}[blue, line width=0.2mm,arrows ={Latex[length=1.2mm]}-{Latex[length=1.2mm]}]
    \draw (0.3,-1.75) -- (0.7,-1.75);
    \draw (0.8,-1.75) -- (1.2,-1.75);
     \end{scope}

     \draw[-{Latex}] (-0.3,-1.75) -- (0,-1.75);
\end{tikzpicture}
           \end{minipage}
	}
	\subfigure[Reward and cost of case 2. ] 
	{\label{fig:case2reward}
 \begin{minipage}[b]{0.46\linewidth}
	\centering
 \includegraphics[height=3.6cm]{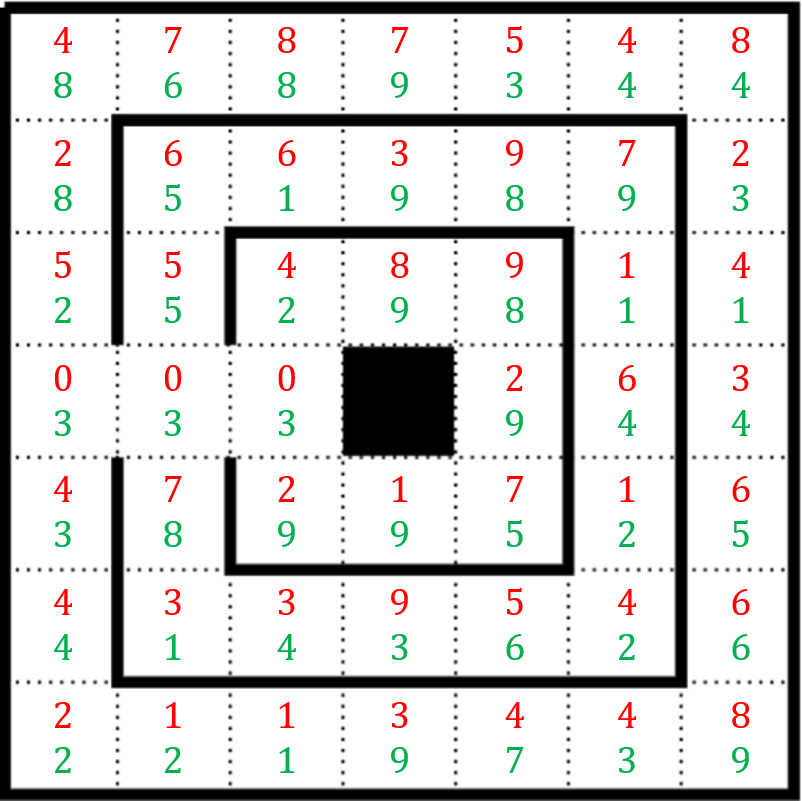}
    \end{minipage}
	}
	
    \caption{Case Study 2.}
   \label{fig:case2}
\end{figure}

\subsection{Discussions}
In this work, the LTL task defines the correctness of the system and is prioritized over the ratio objective, i.e., the ratio objective should be optimized subject to the constraint that the LTL task is satisfied w.p.1. The LTL task can be viewed as a generalized concept of ``safety".  Unlike traditional safety tasks, which typically focus on avoiding obstacles, the LTL task further requires that ``good" outcomes occur infinitely often. For instance, in Case Study 1, the LTL task $\varphi_1$ consists of two components: visiting $d$ infinitely often ($\square \lozenge d$) and never visiting $b$ ($\square \neg b$).

When the LTL formula is simple and the definition of the ratio objective depends on the LTL formula, it is relatively straightforward to avoid conflicts between the LTL task and the ratio objective. For example, in Case Study 1, the ratio objective optimizes the time cost of each visit to $d$ and is compatible with $\varphi_1$.
However, as the complexity of the LTL task increases, avoiding conflicts between the LTL task and the ratio objective becomes more challenging. For instance, in Case Study 1, the task $\varphi_2$ further requires visiting $c$ infinitely often ($\square \lozenge c$), introducing a conflict between the ratio objective and the task $\varphi_2$. One might consider designing a ratio objective that takes the infinite visits to both $d$ and $c$ into account, in order to avoid the conflict. However, in many cases, the quantitative objective is influenced by human preferences. For example, when defining the reward function for $\varphi_2$ in Case Study 1, the designer may prefer the robot to visit $d$ more frequently than $c$, as the robot’s primary task is to carry items, and a higher reward may be assigned to visiting $d$. This may lead to a potential conflict, as the optimal efficiency policy might prioritize visiting $d$ to maximize efficiency and fail to visit $c$ infinitely often.

Furthermore, when additional quantitative objectives, such as the one in Case Study 2, are introduced, designing a suitable ratio objective becomes even more difficult. In Case Study 2, simply adding a sufficiently large reward to the accepting state of the LTL task changes the optimal efficiency policy to one without conflict. However, this policy becomes fixed and does not adapt, even if additional rewards are introduced. Moreover, the true value of the quantitative objective, such as the transportation spare parts amount per time instant in Case Study 2, may become undesirable under this new policy.

From the discussion above, it is clear that conflicts between the LTL task and the quantitative objective are difficult to avoid in general. However, the algorithm proposed in this work guarantees that, even when the two objectives conflict, a conflict-free $\epsilon$-optimal policy is synthesized.

\section{Conclusion}\label{sec:con}
In this paper, we addressed the challenge of maximizing the long-run efficiency of control policies for Markov decision processes, which are characterized by the reward-to-cost ratio, while ensuring that the linear temporal logic task is achieved with probability one. Our results demonstrated that, by exploring stationary policies, it is possible to achieve $\epsilon$-optimality for any threshold value $\epsilon$. Our approach was based on the perturbation analysis technique, originally developed for the classical long-run average reward optimization problem. We extended this technique to the context of long-run efficiency optimization and derived a general formula. Our work not only expanded the theory of perturbation analysis but also highlighted its conceptual simplicity and effectiveness in solving MDPs with both qualitative and quantitative tasks.
In future research, we plan to further investigate how to formulate and solve the multi-objective optimization problem that balances efficiency performance with the visiting frequency of accepting states. 
  
\appendices
\section{Linear programming to solve average reward maximization}\label{app:program}

$\alpha \in \mathbb{R}^{|S|}$ in \eqref{opt2:con6} satisfies that $\alpha(s)>0$ and $\sum_{s \in s} \alpha(s)=1$. The intuitions of the linear program are as follows. 
 The decision variables are $x(s,a)$ and $y(s,a)$ for each state-action pair $s \in S$ and  $a \in A(s)$ in Equation~\eqref{opt2:con7}.
$x(s,a)$ represents steady probability of occupying state $s$ and choosing action $a$, and $y(s,a)$ represents the deviation value at state $s$ and choosing the action $a$.
In Equations~\eqref{opt2:con1} and~\eqref{opt2:con2}, variables $\gamma(s)$ and $\eta(t,s)$ are function of $x(s,a)$ representing the probability of occupying state $s$ and the probability of reaching from states $s$ to $t$, respectively. The variables $\lambda(s)$ and $\zeta(t,s)$ in Equation~\eqref{opt2:con3} and~\eqref{opt2:con4} are function of $y(s,a)$ similar to $\gamma(s)$ and $\eta(t,s)$, respectively. 
Then Equations~\eqref{opt2:con5} and~\eqref{opt2:con6} are constraints for probability flow of stationary distribution and deviation value.  
Finally, objective Equation~\eqref{eq:ratioobjective} compute the average reward for corresponding MC.

\begin{align}
 &\max_{x(s,a),y(s,a)} \quad \sum_{s \in S} \sum_{a \in A(s)} x(s,a)\mathtt{R}_{K}(s,a)   \label{opt2:obj} \\
\!\!\!\!\!\!\!\!\text{s.t. } \ \
&\gamma(s) = \sum_{a \in A(s)} x(s,a), \forall s \in S  \label{opt2:con1}\\ 
&\eta(t,s) = \sum_{a \in A(t)} P(s|t,a) x(t,a), \forall s \in S  \label{opt2:con2}\\ 
&\lambda(s) = \sum_{a \in A(s)} y(s,a), \forall s \in S  \label{opt2:con3}\\
&\zeta(t,s) = \sum_{a \in A(t)}P(s|t,a) y(t,a), \forall s \in S  \label{opt2:con4}\\
&\gamma(s) = \sum_{t \in S} \eta(t,s), \forall s \in S  \label{opt2:con5}\\
&\gamma(s)+\lambda(s) = \sum_{t \in S} \zeta(t,s)+\alpha(s), \forall s \in S  \label{opt2:con6}\\
&x(s,a) \geq 0, y(s,a) \geq 0, \forall s \in S ,\forall a \in A(s)   \label{opt2:con7}
\end{align} 

Let the optimal solution of linear program be $x^{*}(s,a)$ and $y^{*}(s,a)$. We define $S^\star=\{s \in S \mid \sum_{a \in A(s)} x^\star(s,a)>0 \}$. We can constructed a policy $\mu^\star_K$ by following equation:
\begin{align}\label{eq:optimalpolicyforK}
	 \mu^\star_K(s,a)= 
		\left\{
		\begin{array}{ll}
			 x^\star(s,a)/\sum_{a \in A(s)}x^\star(s,a)    &  \text{if } s \in S^\star\\
			 y^\star(s,a)/\sum_{a \in A(s)}y^\star(s,a)  & \text{otherwise.}
		\end{array}
		\right.  \nonumber
\end{align}

\section{Auxiliary Results and Proofs} \label{appendix:proof}
Let $\Phi:\Omega \to \mathbb{R}$ be a pay-off function. The expected pay-off initial from $s \in S$ under policy $\mu \in \Pi_\M$ is $E^\mu_s[\Phi]$. A policy $\mu' \in \Pi_\M$ is optimal w.r.t. pay-off $\Phi$ if $\forall s \in S$, $E^{\mu'}_s[\Phi]=\sup_{\mu \in \Pi_\M} E^\mu_s[\Phi]$.
We say $\Phi$ is \emph{prefix-independent} if for $\omega = s_0a_0s_1a_1\dots \in \Omega$, we have $\omega_n = s_na_ns_{n+1}a_{n+1}\dots \in \Omega$ satisfying $\Phi(\omega)=\Phi(\omega_n)$ for any $n \geq 0$. $\Phi$ is said to be \emph{submixing} if for $\omega = s_0a_0s_1a_1\dots$, $\omega_1 = s_0a_0 s_2 a_2s_4a_4\dots$ and $\omega_2 = s_1a_1s_3a_3s_5a_5\dots\in \Omega$, we have
\[
\Phi(\omega) \leq \max\{ \Phi(\omega_1), \Phi(\omega_2) \}.
\]
Let $\Pi^{SD}_{\M} \subseteq \Pi^{S}_\M$ be stationary deterministic policies set such that for $\mu \in \Pi^{SD}_{\M}$ and $s \in S$, there exists $a \in A(s)$ satisfying $\mu(s,a)=1$.
We now prove the existence of optimal efficiency stationary deterministic policy.
\begin{mycla}\label{cla:ratiooptpolicyexist}
    Given MDP $\M=(S,s_0,A,P,\mathcal{AP},\ell,Acc)$, reward function $\mathtt{R}:S\times A \to 
\mathbb{R}$ and cost function $\mathtt{C}:S\times A\to \mathbb{R}_+$, there exists optimal efficiency policy $\mu^\star \in \Pi^{SD}_\M$, i.e., $\forall s\in S$, $J^{\mu^\star}(s,\mathtt{R},\mathtt{C})=J(s,\mathtt{R},\mathtt{C},\Pi_\M)$.
\end{mycla}
\begin{proof}
    Define the pay-off function $\Phi: \Omega \to \mathbb{R}$ such that for $\omega=s_0a_0s_1a_1\dots \in \Omega$,
    \begin{equation}\label{eq:payoff-def}
        \Phi(\omega)=\liminf_{n\to +\infty}\frac{\sum_{i=0}^n\mathtt{R}(s_i,a_i)}{\sum_{i=0}^n\mathtt{C}(s_i,a_i)}.
    \end{equation}
We now prove that pay-off function \eqref{eq:payoff-def} is prefix-independent and submixing. For $\omega = s_0a_0s_1a_1\dots \in \Omega$, let $\textsf{pre}(\mathtt{R},m)=\sum_{i=0}^{m-1} \mathtt{R}(s_i,a_i) $, $\textsf{suf}(\mathtt{R},m,n)=\sum_{i=m}^{n} \mathtt{R}(s_i,a_i) $, $\textsf{pre}(\mathtt{C},m)=\sum_{i=0}^{m-1} \mathtt{C}(s_i,a_i)$ and $\textsf{suf}(\mathtt{C},m,n)=\sum_{i=m}^{n} \mathtt{C}(s_i,a_i)$. Since $\mathtt{C}$ is a positive function, we have $\lim_{n\to \infty}\textsf{suf}(\mathtt{C},m,n) = +\infty$. Then for $\omega=s_0a_0s_1a_1\dots \in \Omega$ and $\omega_m=s_ma_ms_{m+1}a_{m+1}\dots \in \Omega$, we have
\begin{equation}
    \begin{aligned}
        & \Phi(\omega)= \liminf_{n\to +\infty} \frac{\sum_{i=0}^{n}\mathtt{R}(s_i,a_i)}{\sum_{i=0}^{n}\mathtt{C}(s_i,a_i)} \\
        =&\liminf_{n\to+\infty} \frac{\textsf{pre}(\mathtt{R},m)+\textsf{suf}(\mathtt{R},m,n)}{\textsf{pre}(\mathtt{C},m)+\textsf{suf}(\mathtt{C},m,n)} \\
        =&\liminf_{n\to +\infty} \left(\frac{\textsf{pre}(\mathtt{R},m)}{\textsf{suf}(\mathtt{C},m,n)}+\frac{\textsf{suf}(\mathtt{R},m,n)}{\textsf{suf}(\mathtt{C},m,n)}\right)/(\frac{\textsf{pre}(\mathtt{C},m)}{\textsf{suf}(\mathtt{C},m,n)}+1)\\
        =& \liminf_{n\to +\infty} \frac{\textsf{suf}(\mathtt{R},m,n)}{\textsf{suf}(\mathtt{C},m,n)}=\Phi(\omega_m).
        \nonumber
    \end{aligned}
\end{equation}
The last equality holds because $\textsf{pre}(\mathtt{R},m)/\textsf{suf}(\mathtt{C},m,n)$ and $\textsf{pre}(\mathtt{C},m)/\textsf{suf}(\mathtt{C},m,n)$ are zero as $n \to +\infty$.
Thus pay-off function $\Phi$ is prefix-independent.

For $c/a$ and $d/b$ such that $a,b>0$, assume that $c/a \geq d/b$. Then $bc\geq ad$. Thus $ac+bc\geq ad+ac$ and we get $(c+d)/(a+b)\leq c/a$. It holds that  $(c+d)/(a+b)\leq \max\{ c/a ,d/b\}$. Then, for $s_0a_0\dots s_{2n-1}a_{2n-1}$, let $a_n=\sum_{i=0}^{n-1} \mathtt{C}(s_{2i},a_{2i})$, $b_n=\sum_{i=0}^{n-1} \mathtt{C}(s_{2i+1},a_{2i+1})$, $c_n=\sum_{i=0}^{n-1} \mathtt{R}(s_{2i},a_{2i})$, $d_n=\sum_{i=0}^{n-1} \mathtt{R}(s_{2i+1},a_{2i+1})$, we have $(c_n+d_n)/(a_n+b_n)\leq \max\{ c_n/a_n ,d_n/b_n\}$. Since above result holds for any $n$, we know that function $\Phi$ is submixing. 

Since $\Phi$ is prefix-independent and submixing, with result in~\cite{gimbert2007pure}, there exists deterministic policy $\mu^\star \in \Pi^{SD}_\M$ such that $\forall s \in S$, $E^{\mu^\star}_s[\Phi]=\sup_{\mu \in \Pi_\M} E^\mu_s[\Phi]$. It means that criterion
\[
 E \left\{  \liminf_{N\to +\infty} \frac{\sum_{i=0}^{N} \mathtt{R}(s_i,a_i)}{\sum_{i=0}^{N} \mathtt{C}(s_i,a_i)} \right\}
\]
has stationary deterministic optimal policy. Then from \cite{bierth1987expected}, the policy $\mu^\star$ also optimizes the criterion in \eqref{eq:ratioobjectdef}, i.e., $\forall s\in S$, $J^{\mu^\star}(s,\mathtt{R},\mathtt{C})=J(s,\mathtt{R},\mathtt{C},\Pi_\M)$. It completes the proof.
\end{proof}
Given stationary policy $\mu \in \Pi^S_\M$, assume that MC $\M^\mu$ has $k$ recurrent class $R_1,R_2,\dots,R_k \subseteq S$. From \cite{vonsynthesizing} we have
\begin{equation} \label{eq:multichainratio}
    J^{\mu}(s,\mathtt{R},\mathtt{C})=\sum_{i=1}^k\textsf{Pr}^{\mu}(s,R_i) J^{\mu}(s(R_i),\mathtt{R},\mathtt{C})
\end{equation}
where $\textsf{Pr}^{\mu}(s,R_i)$ is the reaching probability in MC $\M^{\mu}$ when initial state is $s$ \cite[Page 759]{baier2008principles} and $s(R_i) \in R_i$ is arbitrary state in $R_i$. Note that every recurrent class is in some MEC. We say a stationary policy $\mu \in \Pi^S_\M$ is \emph{regular} if for each MEC $(\S,\A) \in \texttt{MEC}(\M)$, one of (a) and (b) holds: (a) All states in $\S$ are transient in $\M^{\mu}$; (b) In MC $\M^{\mu}$, only one recurrent class $R \subseteq \S$ is in MEC $(\S,\A)$ and states in $\S \setminus R$ will reach $R$ eventually.
For regular policy $\mu$ and MEC $(\S,\A)$ such that (b) holds, we have $\textsf{Pr}^{\mu}_{\texttt{R}}(\S,\A) = \textsf{Pr}^\mu(s_0,R_j)$ where $\textsf{Pr}^\mu_\texttt{R}(\S,\A)$ is defined in \eqref{eq:stayingforeverinMEC}.
From \eqref{eq:multichainratio}, we have
\begin{equation}\label{eq:MECreachingratio}
\begin{aligned}
    J^{\mu}(s_0,\mathtt{R},\mathtt{C}) & = \sum_{(\S,\A) \in \texttt{MEC}(\M)} \textsf{Pr}^{\mu}_\texttt{R}(\S,\A) J^{\mu}(s_{(\S,\A)},\mathtt{R},\mathtt{C}) \\
    & = \sum_{(\S,\A) \in \texttt{MEC}(\M)} \textsf{Pr}^{\mu}_\texttt{R}(\S,\A) \frac{W^{\mu}(s_{(\S,\A)},\mathtt{R})}{W^{\mu}(s_{(\S,\A)},\mathtt{C})},
\end{aligned}
\end{equation}
such that $s_{(\S,\A)} \in \S$ can be any state in $\S$. Let $\mu^\star$ be the optimal efficiency policy w.r.t. $\mathtt{R}$ and $\mathtt{C}$. We now prove that it is without loss of generality to assume that $\mu^\star$ is regular.
\begin{mycla}\label{cla:regularoptimalpolicy}
Given MDP $\M$, reward $\mathtt{R}$ and cost $\mathtt{C}$. We can find a regular policy $\mu^\star \in \Pi^{SD}_\M$ which is optimal deterministic stationary policy w.r.t. $\mathtt{R}$ and $\mathtt{C}$, i.e.,
\begin{equation}
    J^{\mu^\star}(s,\mathtt{R} ,\mathtt{C})=J(s,\mathtt{R} ,\mathtt{C},\Pi_\M), \forall s\in S. \nonumber
\end{equation}
\end{mycla}
\begin{proof}
   By \cite[Thm 8.3.2]{puterman}, for MEC $(\S,\A) \in \texttt{MEC}(\M)$,
   \begin{equation} \label{eq:comequationratio}
J(s,\mathtt{R},\mathtt{C},\Pi_\M)=J(s',\mathtt{R},\mathtt{C},\Pi_\M), \forall s,s' \in \S.
   \end{equation}
Assume that in MC $\M^{\mu^\star}$ there are several recurrent classes in $(\S,\A)$. Let $s, s' \in \S$ be two states in different recurrent classes. Since $J^{\mu^\star}(s,\mathtt{R},\mathtt{C})=J^{\mu^\star}(s',\mathtt{R},\mathtt{C})$ and \eqref{eq:multichainratio}, these two recurrent classes achieve same efficiency value. Thus we can modify $\mu^\star$ such that all states in $\S$ will reach only one of these recurrent classes eventually and achieve same efficiency value. Thus, it is without loss of generality to assume that each MEC has at most one recurrent class. Assume that some MEC $(\S,\A)$ has one recurrent class and $s \in \S$ will leave the MEC eventually with non-zero probability. By \eqref{eq:comequationratio}, we can modify $\mu^\star$ such that $s$ will stay in $(\S,\A)$ forever w.p.1 and achieve same efficiency value. This completes the proof.
\end{proof}
By result of Claim~\ref{cla:regularoptimalpolicy}, we assume that any optimal efficiency policy is regular in this work. For a regular policy $\mu \in \Pi^S_\M$ and an MEC $(\S,\A) \in \texttt{MEC}(\M)$ that contains recurrent class in MC $\M^\mu$, if we modify $\mu$ to $\mu'$ over $(\S,\A)$ such that $\mu'$ is also a regular policy but the recurrent class of $(\S,\A)$ in MC $\M^{\mu'}$ is different, the probability of staying forever in MEC $(\S,\A)$ are same in MC $\M^{\mu}$ and $\M^{\mu'}$, i.e.,
\begin{equation}\label{eq:samestayingprobability}
    \textsf{Pr}^{\mu}_\texttt{R}(\S,\A) = \textsf{Pr}^{\mu'}_\texttt{R}(\S,\A).
\end{equation}
For $\mu \in \Pi_{\M}$, end component $(\hat{\S},\hat{\A})$, 
let
\begin{equation}\label{eq:reachingprobability}
\textsf{Pr}^\mu(\hat{\S},\hat{\A})=\textsf{Pr}^{\mu}_{\M}(\{ \omega \in \Omega \mid \texttt{limit}(\omega) = (\hat{\S},\hat{\A}) \})
\end{equation}
be the probability of sample path which just visits all state-action pairs in EC $(\hat{\S},\hat{\A})$ infinitely often.

For $\mu \in \Pi^S_\M$ and its limit transition matrix $ (\mathbb{P}^{\mu})^\star$, we define $p(\mu)=\min \{ (\mathbb{P}^{\mu})^\star_{s,t} \mid s,t \in S \wedge (\mathbb{P}^{\mu})^\star_{s,t}>0 \}$ the smallest non-zero limit probability under policy $\mu \in \Pi^S_\M$. We define
\begin{equation}\label{eq:minlimitamongdeter}
    \hat{p}=\min\{ p(\mu) \mid \mu \in \Pi^{SD}_\M \}
\end{equation}
 the smallest non-zero limit probability among stationary deterministic policies. Since $\Pi^{SD}_\M$ is finite, the minimum operation in \eqref{eq:minlimitamongdeter} is well-defined. Now suppose that $\M$ has $n$ MAECs, i.e., 
$\texttt{MAEC}=\{ (\S_1,\A_1),(\S_2,\A_2),\dots,(\S_n,\A_n) \}$. We define $\hat{r}= \max_{s\in S, a \in A} |\mathtt{R}(s,a)|$, $\hat{c} = \min_{s\in S, a \in A} \mathtt{C}(s,a)$ and $\Bar{c}=\max_{s \in S, a \in A} \mathtt{C}(s,a)$. Then we define reward function $\hat{\mathtt{R}}$:
\begin{align}\label{eq:modifiedreward}
	 \hat{\mathtt{R}}(s,a)= 
		\left\{
		\begin{array}{ll}
  \mathtt{R}(s,a)    &  \text{if } s \in \S_i \wedge a \in \A_i(s)
			  \\
    -\frac{(1+ \frac{\Bar{c}}{\hat{c}})\hat{r}}{\hat{p}}    &  \text{otherwise. }
		\end{array}
		\right. 
\end{align}
Then we prove Claim~\ref{claim:hatrewardupperbound} and Claim~\ref{claim:boundforLTLstate} to characterize the optimal efficiency under reward $\hat{\mathtt{R}}$ and cost $\mathtt{C}$, i.e., $J(s,\hat{\mathtt{R}} ,\mathtt{C},\Pi_\M)$.
\begin{mycla} \label{claim:hatrewardupperbound}
Given original reward $\mathtt{R}$ and cost $\mathtt{C}$, the modified reward $\hat{\mathtt{R}}$ in \eqref{eq:modifiedreward}, and initial state $s_0$, we have
\begin{equation}\label{eq:stationaryhalf}
     J(s_0,\hat{\mathtt{R}} ,\mathtt{C},\Pi_\M) \geq J(s_0,\mathtt{R},\mathtt{C},\Pi_\M^\varphi).
\end{equation}
\end{mycla}
\begin{proof}
    For $\mu \in \Pi^\varphi_\M$, $(\hat{\S},\hat{\A}) \in \texttt{AEC}(\M)$ and $\textsf{Pr}^\mu(\hat{\S},\hat{\A})$ in \eqref{eq:reachingprobability}, we have $\sum_{(\hat{\S},\hat{\A}) \in \texttt{AEC}(\M)} \textsf{Pr}^\mu(\hat{\S},\hat{\A})=1$ from \cite[Thm 10.122]{baier2008principles}. Thus the  state-action pairs visited infinitely often are in MAECs with probability 1. Since 1) the objective value in \eqref{eq:ratioobjectdef} is only dependent on state-action pairs that appear infinitely often, and 2) $\mathtt{R}$ and $\hat{\mathtt{R}}$ are same over MAECs, we have $J^{\mu}(s_0,\hat{\mathtt{R}},\mathtt{C}) = J^{\mu}(s_0,\mathtt{R},\mathtt{C})$.
Since $\mu$ is arbitrary policy in $\Pi^\varphi_\M$, we complete the proof.
\end{proof}
\begin{mycla}\label{claim:boundforLTLstate}
For $s\in S$, we have $J(s,\hat{\mathtt{R}},\mathtt{C},\Pi_\M)\geq -\frac{\hat{r}}{\hat{c}}$.
\end{mycla}
\begin{proof}
    Let $(\S,\A) \in \texttt{MAEC}(\M)$ and $R \subseteq \S$ be a recurrent class. From \eqref{eq:ratioobjective} the ratio objective value initial from $R$ is
\begin{equation}\label{eq:rewardboundforeasy}
    \begin{aligned}
        & \frac{\sum_{s \in R} \sum_{a \in \A(s)} \pi(s) \mu(s,a) \hat{\mathtt{R}}(s,a)}{\sum_{s \in R} \sum_{a \in \A(s)} \pi(s) \mu(s,a) \mathtt{C}(s,a)} \\
       \geq & \frac{\sum_{s \in R} \sum_{a \in \A(s)} \pi(s) \mu(s,a) -\hat{r}}{\sum_{s \in R} \sum_{a \in \A(s)} \pi(s) \mu(s,a) \hat{c}} \\
       = & -\frac{\hat{r}}{\hat{c}},
    \end{aligned}
\end{equation}
where $\pi$ is the limit distribution over the recurrent class $R$. The inequality holds since $\hat{\mathtt{R}}(s,a) \geq -\hat{r}$ for $s \in \S, a\in \A(s)$. Since we assumed that the LTL task can be finished w.p.1 regardless of initial state, then initial from each $s\in S$, there exists policy under which MDP will stay in MAECs forever w.p.1. Combining with \eqref{eq:rewardboundforeasy} we complete the proof. 
\end{proof}
We finally prove that under optimal policy of efficiency w.r.t. $\hat{\mathtt{R}}$ and $\mathtt{C}$, the recurrent states are in MAECs. 
\begin{mycla}\label{claim:deterpolicyproperty}
We denote by $\mu^\star \in \Pi^{SD}_\M$ the optimal deterministic stationary policy such that
\begin{equation}\label{eq:optimaldeterdef}
    J^{\mu^\star}(s,\hat{\mathtt{R}} ,\mathtt{C})=J(s,\hat{\mathtt{R}} ,\mathtt{C},\Pi_\M), \forall s\in S.
\end{equation}
The existence of policy $\mu^\star$ comes from Claim~\ref{cla:ratiooptpolicyexist}. Then following statements hold without loss of generality:
    \begin{enumerate}
        \item $s$ is transient in $\M^{\mu^\star}$ if $\mu^\star(s,a)=1$, $\hat{\mathtt{R}}(s,a)\neq \mathtt{R}(s,a)$.
    \item The efficiency values are same with rewards $\hat{\mathtt{R}}$ and $\mathtt{R}$, i.e., \begin{equation}\label{eq:starissame}
    J^{\mu^\star}(s_0,\hat{\mathtt{R}},\mathtt{C})=J^{\mu^\star}(s_0,\mathtt{R},\mathtt{C}).
\end{equation}
    \end{enumerate}
\end{mycla}
\begin{proof}
We prove 1) by contradiction. Assume that $\mu^\star(\hat{s},a)=1$, $\hat{\mathtt{R}}(\hat{s},a)\neq \mathtt{R}(\hat{s},a)$ and $\hat{R} \subseteq S$ is the recurrent class $\hat{s}$ belongs to. Then from \eqref{eq:ratioobjective} we have
\begin{equation}
    \begin{aligned}
        & J^{\mu^\star}(\hat{s},\hat{\mathtt{R}} ,\mathtt{C})\\
        =& \frac{ -\pi_{\hat{s}}(\hat{s})\frac{(1+\frac{\Bar{c}}{\hat{c}})\hat{r}}{\hat{p}} + \sum_{s \in \hat{R} \setminus \{ \hat{s}\} } \sum_{a \in A(s)} \pi_{\hat{s}}(s) \mu^\star(s,a) \hat{\mathtt{R}}(s,a)}{\sum_{s \in \hat{R}} \sum_{a \in A(s)} \pi_{\hat{s}}(s) \mu^\star(s,a) \mathtt{C}(s,a)} \\
       \leq & \frac{-(1+\frac{\Bar{c}}{\hat{c}})\hat{r}+\sum_{s \in \hat{R}\setminus \{s\}} \sum_{a \in A(s)} \pi_{\hat{s}}(s) \mu^\star(s,a) \hat{r}}{\sum_{s \in \hat{R}} \sum_{a \in A(s)} \pi_{\hat{s}}(s) \mu^\star(s,a) \mathtt{C}(s,a)} \\
       < & \frac{-\frac{\Bar{c}}{\hat{c}}\hat{r}}{\Bar{c}} =-\frac{\hat{r}}{\hat{c}},
        \nonumber
    \end{aligned}
\end{equation}
where $\pi_{\hat{s}}$ is the limit distribution of MC $\M^{\mu^\star}$, i.e., the row of state $\hat{s}$ of limit transition matrix $(\mathbb{P}^{\mu^\star})^\star$. Since $\pi_{\hat{s}}(\hat{s})>0$, from definition of $\hat{p}$ in \eqref{eq:minlimitamongdeter}, we have $\pi_{\hat{s}}(\hat{s})>\hat{p}$. Since the denominator is positive and $\hat{\mathtt{R}}(s,a)\leq \hat{r}$, we get first inequality. Since $\sum_{s \in R\setminus \{s\}} \sum_{a \in A(s)} \pi_{\hat{s}}(s) \mu^\star(s,a) \hat{r}=(1-\pi_{\hat{s}}(\hat{s}))\hat{r}<\hat{r}$ and numerator is negative, we know second inequality holds. It violates the result of Claim~\ref{claim:boundforLTLstate}. Thus 1) holds.

To prove 2), from 1), if state $s$ is recurrent in MC $\M^{\mu^\star}$ and $\mu^\star(s,a)=1$, then $\hat{\mathtt{R}}(s,a)=\mathtt{R}(s,a)$. Thus \eqref{eq:starissame} holds. 
\end{proof}

\bibliographystyle{plain}
\bibliography{main}

\end{document}